\documentclass[conference]{IEEEtran}
\IEEEoverridecommandlockouts
\usepackage{tikz}
\usepackage{amsmath}
\usepackage{filecontents}
\usepackage{xspace}
\usepackage{algorithm}
\usepackage{algorithmic}
\usepackage{array}
\usepackage{bm}
\usepackage{xspace}
\usepackage{amsfonts}
\usepackage{multirow}
\usepackage{bm}
\usepackage{amsmath,amsfonts}
\usepackage{xcolor,soul}
\usepackage{amsthm}
\usepackage{comment}
\usepackage{array}
\newcommand{\methodName}{\textsc{FedSpy-LLM}\xspace}
\newtheorem{theorem}{Theorem}
\newtheorem{lemma}{Lemma}

\usepackage{colortbl} % For coloring rows/columns
\usepackage{lipsum}   % For dummy text
\usepackage{graphicx} % For resizebox
\usepackage{cellspace} % For adding padding inside table cells
\usepackage{enumitem}
\usepackage{booktabs}
\usepackage{makecell}
\usepackage{threeparttable} % For adding footnotes to tables

\usepackage{color,soul}
\definecolor{highlight1}{rgb}{1.0, 1.0, 0.0} % Yellow
\definecolor{highlight2}{rgb}{0.8, 1.0, 0.8} % Light Green

\newcommand{\hlone}[1]{\sethlcolor{highlight1}\hl{#1}}
\newcommand{\hltwo}[1]{\sethlcolor{highlight2}\hl{#1}}
\newcommand*\circled[1]{\tikz[baseline=(char.base)]{
            \node[shape=circle,fill,inner sep=1pt] (char) {\textcolor{white}{#1}};}}

\def\BibTeX{{\rm B\kern-.05em{\sc i\kern-.025em b}\kern-.08em
    T\kern-.1667em\lower.7ex\hbox{E}\kern-.125emX}}
\begin{document}

\title{\methodName: Towards Scalable and Generalizable Data Reconstruction Attacks from Gradients on LLMs}
\author{\IEEEauthorblockN{Syed Irfan Ali Meerza\IEEEauthorrefmark{1},
Feiyi Wang\IEEEauthorrefmark{3}, Jian Liu\IEEEauthorrefmark{1}\IEEEauthorrefmark{2}}
\IEEEauthorblockA{\IEEEauthorrefmark{1}University of Tennessee, Knoxville, TN, USA}
\IEEEauthorblockA{\IEEEauthorrefmark{3}Oak Ridge National Laboratory,
Oak Ridge, USA}
\IEEEauthorblockA{\IEEEauthorrefmark{2}University of Georgia, Athens, GA, USA}
\IEEEauthorblockA{Emails: smeerza@vols.utk.edu; fwang2@ornl.gov; jianliu@uga.edu}}
\maketitle

\begin{abstract}
Given the growing reliance on private data in training Large Language Models (LLMs), Federated Learning (FL) combined with Parameter-Efficient Fine-Tuning (PEFT) has garnered significant attention for enhancing privacy and efficiency.
Despite FL's privacy benefits, prior studies have shown that private data can still be extracted from shared gradients. However, these studies, mainly on full-parameter model training, are limited to reconstructing small batches and short input sequences, and specific model architectures, such as encoder-based or decoder-based models.
The reconstruction quality will become even worse when dealing with gradients from PEFT methods. 
To fully understand the practical attack surface of federated LLMs, this paper proposes \methodName, a scalable and generalizable data reconstruction attack designed to reconstruct training data with larger batch sizes and longer sequences while generalizing across diverse model architectures, even when PEFT methods are deployed for training.
At the core of \methodName is a novel gradient decomposition strategy that exploits the rank deficiency and subspace structure of gradients, enabling efficient token extraction while preserving key signal components at scale. This approach further mitigates the reconstruction challenges introduced by PEFT's substantial null space, ensuring robustness across encoder-based, decoder-based, and encoder-decoder model architectures. Additionally, by iteratively aligning each token’s partial-sequence gradient with the full-sequence gradient, \methodName ensures accurate token ordering in reconstructed sequences. 
Extensive evaluations demonstrate that \methodName consistently outperforms prior attacks and maintains strong reconstruction quality under realistic and challenging settings, revealing a broader and more severe privacy risk landscape in federated LLMs. These findings underscore the urgent need for more robust privacy-preserving techniques in future FL systems.
\end{abstract}

%\begin{IEEEkeywords}
%component, formatting, style, styling, insert
%\end{IEEEkeywords}

\section{Introduction}

\begin{table*}[t]
\centering
\caption{Comparison of state-of-the-art data reconstruction attacks on LLMs.}
\label{tab: comparision}
\renewcommand\arraystretch{1.2}

\resizebox{\linewidth}{!}{%
\begin{tabular}{|c|c|c|c|c|c|c|c|c|c|}
\hline
\multirow{2}{*}{\textbf{Attack}} &
\multirow{2}{*}{\textbf{Year}} &
\multirow{2}{*}{\textbf{Attack Method}} &
\multirow{2}{*}{\makecell{\textbf{Target Models}\textsuperscript{1}}} &
\multirow{2}{*}{\textbf{Max Batch Size}} &
\multirow{2}{*}{\textbf{PEFT}} &
\multicolumn{3}{c|}{\textbf{Effectiveness Against}\textsuperscript{2}} &
\multirow{2}{*}{\textbf{Time Complexity}} \\

& & & & & & \textbf{Encoder-based} & \textbf{Decoder-based} & \textbf{Encoder-Decoder} & \\ \hline

DLG~\cite{zhu2019deep} &
2019 &
Optimization &
Encoder-based &
1 &
$\times$ &
High & Medium & Low &
$\mathcal{O}(B \cdot P \cdot d^2)$ \\ \hline

TAG~\cite{dengtaggradient} &
2021 &
Optimization &
Encoder-based &
1 &
$\times$ &
High & Medium & Low &
$\mathcal{O}(B \cdot P \cdot d^2)$ \\ \hline

LAMP~\cite{dimitrov2022lamp} &
2022 &
Optimization &
Encoder-based &
4 &
$\times$ &
High & Medium & Low &
$\mathcal{O}(B \cdot P \cdot d^2)$ \\ \hline

BGP~\cite{li2023beyond} &
2022 &
Optimization + Search &
Encoder-based &
8 &
$\times$ &
High & N/A\textsuperscript{3} & Low &
$\mathcal{O}(B \cdot d^3)$ \\ \hline

FILM~\cite{gupta2022recovering} &
2022 &
Search &
Decoder-based &
128 &
$\times$ &
Low & High & Medium &
$\mathcal{O}(B \cdot |V| \cdot d)$ \\ \hline

\multirow{2}{*}{DAGER~\cite{petrov2024dager}} &
\multirow{2}{*}{2024} &
\multirow{2}{*}{Search} &
\multirow{2}{*}{\makecell{Encoder-based, Decoder-based}} &
\multirow{2}{*}{128} &
\multirow{2}{*}{$\checkmark$\textsuperscript{4}} &
\multirow{2}{*}{Medium} & \multirow{2}{*}{High} & \multirow{2}{*}{Low} &
$\mathcal{O}(B \cdot P \cdot d^2)$ (Decoder) \\

&&&&&&&&& $\mathcal{O}(B \cdot |V|^P \cdot d^2)$ (Encoder) \\ \hline

\textbf{\methodName} &
2025 &
Optimization &
\makecell{Encoder-based, Decoder-based,\\ Encoder-Decoder} &
128 &
$\checkmark$ &
High & High & High &
$\mathcal{O}(B \cdot P \cdot d^2)$ \\ \hline

\end{tabular}%
}

\vspace{1mm}
\parbox{\linewidth}{\footnotesize
\textsuperscript{1}\textit{Target Models} refer to the types of models evaluated in the original papers, encompassing encoder-based, decoder-based, and encoder-decoder architectures.\\
\textsuperscript{2}\textit{Effectiveness Against} refers to the relative attack performance across different target model architectures.\\
\textsuperscript{3} BGP specifically targets the Pooler and Classifier layers, which are architectural components unique to encoder-based models.\\
\textsuperscript{4} Low-rank PEFT updates make gradients sparse, reducing the effectiveness of search-based methods in distinguishing true embeddings from noise.
}
\vspace{-2mm}
\end{table*}

Language models (LMs) are statistical models that assign probabilities to sequences of words, forming the backbone of many natural language processing tasks. Modern neural-network-based LMs are trained on vast datasets, often nearing a terabyte of text data~\cite{brown2020language, raffel2020exploring}. 
However, the vast amounts of data required for training these models often include sensitive information, raising significant privacy concerns when sharing data with third parties~\cite{das2024security}.
Federated Learning (FL) offers a
promising solution to these privacy issues by enabling multiple
parties to train a model collaboratively without sharing their
private data~\cite{mcmahan2017communication}.  
%\textcolor{red}{While federated learning reduces direct data exposure, it does not by itself provide formal privacy guarantees, and additional mechanisms such as secure aggregation or differential privacy are required for rigorous privacy protection}~\cite{mcmahan2017communication}. 
Instead of exchanging raw data, participants share only the gradients computed on their local data with a central server. This approach has proven effective in fine-tuning Large Language Models (LLMs) in many privacy-sensitive areas, such as law~\cite{zhang2023fedlegal} and healthcare~\cite{sadilek2021privacy}\footnote{Note that, while FL reduces direct data exposure, it does not by itself provide formal privacy guarantees, and additional mechanisms such as secure aggregation or differential privacy are required for rigorous privacy protection.}.

Despite its advantages, FL still poses privacy risks. Existing studies have demonstrated that FL can still leak private information from shared gradients.  In the image domain, data reconstruction attacks can reconstruct detailed images from gradients, compromising visual data privacy~\cite{geiping2020inverting, yin2021see, li2022auditing}. Similar vulnerabilities exist for audio data, where adversaries can reconstruct original recordings~\cite{ovi2024gradient, li2023speech}. 
While LLMs trained with FL are also susceptible to data reconstruction attacks (e.g.,~\cite{dimitrov2022lamp, dengtaggradient, gupta2022recovering, li2023beyond}), 
the discrete nature of text data and the high dimensionality of the search space complicate these attacks, often resulting in only partial recovery of small batches and short sequences. Additionally, the architectural characteristics of LLMs, such as token embeddings and attention mechanisms, obscure the direct relationship between gradients and the original data, making reconstructions less accurate and coherent. Unlike traditional classification tasks, LLMs are trained with an autoregressive next-token prediction objective, where gradients entangle information from multiple tokens and their positions simultaneously. Consequently, reconstruction requires recovering not only the token identities but also their correct sequential order, a challenge that \methodName explicitly addresses through gradient subspace-based token recovery followed by sequence order calibration.

Moreover, the substantial size of LLMs presents additional challenges. Directly training or fine-tuning these models for downstream tasks incurs significant communication overhead and strains the storage and computational resources of participating devices. A promising solution to these challenges is Parameter-Efficient Fine-Tuning (PEFT)~\cite{zhang2023fedpetuning}, which involves updating a small number of additional parameters while keeping the vast majority of the model's original weights frozen (e.g. LoRA~\cite{hu2021lora} and Adapters~\cite{houlsby2019parameter}). 
To better accommodate the decentralized nature of FL, various PEFT methods have been proposed, such as SLoRA~\cite{babakniya2023slora} and FedAdapter~\cite{cai2022fedadapter}, each tailored to preserve the efficiency and effectiveness of employing PEFT in FL settings.
By reducing the volume of gradients exchanged, PEFT not only alleviates resource constraints but also narrows the attack surface, making it more difficult for adversaries to exploit vulnerabilities.

\textbf{Prior Data Reconstruction Attacks on LLMs.}
Data reconstruction attacks, first proposed in~\cite{zhu2019deep}, demonstrated that an adversary could fully restore a client’s private data samples by finding the optimal pair of dummy input and label that best matches the shared gradients. While significant progress has been made in image reconstruction, text-based data reconstruction attacks remain relatively underdeveloped.
In FL, where discrete text data trains language models, various methods have been developed to recover private text from shared gradients. Early methods, DLG~\cite{zhu2019deep} and TAG~\cite{dengtaggradient}, adapt gradient matching techniques to optimize dummy inputs and labels, effectively recovering private tokens. However, they encountered scalability challenges, particularly with larger sequences or batch sizes. To improve reconstructions, LAMP~\cite{dimitrov2022lamp} incorporates language model priors, combining continuous optimization to minimize gradient reconstruction loss with discrete search methods that refine token sequences based on perplexity scores, ensuring semantically plausible outputs. FILM~\cite{gupta2022recovering} employs beam search to systematically explore token sequences, achieving higher recovery accuracy, albeit at a higher computational cost. More recently, DAGER~\cite{petrov2024dager} leverages the low-rank structure of transformer gradients, efficiently identifying token embeddings within the gradient's span to enable exact token recovery, particularly for larger batches and longer sequences but it faces challenges with exhaustive search heuristics in encoder-based models, particularly when dealing with extremely large vocabularies or datasets.
Table~\ref{tab: comparision} provides a comprehensive comparison of the scalability and generalization performance of existing attacks. 

\textbf{Scalability and Generalization Issues of Prior Attacks.}
Despite the existing efforts, data reconstruction attacks on LLMs face significant challenges in practical FL settings: 
\setlist[itemize]{leftmargin=*}
\begin{itemize}
    \item \circled{1} \textit{\underline{Batch Size and Context Length}:}  
    In LLMs, gradients are computed over full sequences, blending information across tokens and making individual token contributions hard to isolate. This, along with the highly non-convex optimization landscape, hampers the effectiveness of gradient-based reconstruction attacks like DLG~\cite{zhu2019deep}, TAG~\cite{dengtaggradient}, LAMP~\cite{dimitrov2022lamp}, and BGP~\cite{li2023beyond}, which only scale to small batches (1–8) and short sequences. Their dependence on full gradient signals limits scalability and practicality.
    In contrast, more recent methods like FILM~\cite{gupta2022recovering} and DAGER~\cite{petrov2024dager} improve scalability (up to batch size 128) by using search-based approaches. However, the trade-off is significant computational overhead, especially when applied to more complex or larger-scale LLM architectures (e.g., encoder-based vs.\ encoder-decoder), where the search space grows exponentially. This is because, in encoder-based models, each token attends to all other tokens bidirectionally, leading to complex gradient interactions that inflate the effective search space. Unlike decoder-only models with autoregressive factorization, encoder models do not provide a natural sequential structure to constrain token prediction. Similarly, in encoder-decoder architectures, the interaction between encoder and decoder layers (especially through cross-attention) introduces additional entanglement, making it harder to evaluate isolated token candidates. As a result, search-based methods like DAGER struggle to identify the correct token among many candidates due to poor alignment between isolated token gradients and full-sequence gradients in these architectures.

    \item \circled{2} \textit{\underline{Reconstruction from PEFT Gradients}:} The use of PEFT methods (e.g.,~\cite{babakniya2023slora,cai2022fedadapter}) complicates the process even further. By modifying only a small subset of parameters, PEFT techniques generate gradients that primarily reflect localized changes rather than full-model updates. This sparsity weakens gradient-based attacks like DLG, TAG, LAMP, and BGP, which rely on rich gradient signals to distinguish token embeddings.
    Although search-based methods like DAGER can recover tokens even under PEFT, 
    the low-rank nature of PEFT updates results in sparse gradients that lack sufficient information~\cite{song2023sparse}. This sparsity increases ambiguity by making it harder to distinguish between the actual token embeddings and other unrelated embeddings that produce similar gradients.

    \item \circled{3} \textit{\underline{Model Architecture}:} 
    Experimental results reveal that many gradient inversion attacks face architectural constraints. Although methods like DLG~\cite{zhu2019deep}, TAG~\cite{dengtaggradient}, and LAMP~\cite{dimitrov2022lamp} are nominally architecture-agnostic, their reliance on gradient matching and focus on encoder-based models limits effectiveness for decoder-based and encoder-decoder architectures, where cross-attention adds complexity. BGP~\cite{li2023beyond}, which leverages the Pooler layer for attacks, demonstrates strong performance against encoder-based models but is less effective for decoder-based models due to the absence of a Pooler layer. FILM~\cite{gupta2022recovering}, which uses word embedding gradients in a search-based framework, excels in decoder-based models but struggles with encoder-based and encoder-decoder setups due to bidirectional context. DAGER~\cite{petrov2024dager} broadens coverage across architectures, but suffers from high computational costs in encoder-decoder models, where the interplay of encoder and decoder states greatly enlarges the search space. These architectural challenges, compounded by optimization difficulties, sparse gradients from PEFT, and limited positional cues, highlight the gap between theoretical attacks and practical feasibility.
   
    \item \circled{4} \textit{\underline{Time Complexity}:} 
    The time complexities of data reconstruction attacks vary significantly due to their reconstruction strategies. Optimization-based approaches such as DLG~\cite{zhu2019deep}, TAG~\cite{dengtaggradient}, and LAMP~\cite{dimitrov2022lamp} exhibit a complexity of $\mathcal{O}(B \cdot P \cdot d^2)$, where $B$ is the batch size, $P$ the sequence length, and $d$ the embedding dimension. These methods rely on iterative gradient matching and are generally efficient, though often less precise for longer sequences or sparse gradients. 
    FILM~\cite{gupta2022recovering} uses beam search, increasing complexity to $\mathcal{O}(B \cdot |V| \cdot d)$, where $|V|$ is the vocabulary size.
    %FILM~\cite{gupta2022recovering}, introduces a beam search mechanism over candidate tokens, resulting in higher overhead dominated by $\mathcal{O}(B \cdot |V| \cdot d)$, where $|V|$ is the vocabulary size. 
    BGP~\cite{li2023beyond} adds an additional analytic stage based on tensor decomposition, bringing the total complexity to $\mathcal{O}(B \cdot d^3)$. In contrast, 
    DAGER~\cite{petrov2024dager}, with exhaustive token search, incurs exponential complexity $\mathcal{O}(B \cdot |V|^P \cdot d^2)$ for encoder-based models. While search-based methods often yield better accuracy, their high computational overhead limits scalability across sequence lengths, batch sizes, and architectures.
    
    \item \circled{5} \textit{\underline{Sequence Order}:}  
    Frozen positional embeddings and gradient averaging during fine-tuning strip away direct cues about token positions. This loss of positional information, compounded by the model’s non-convex optimization landscape, makes it difficult to recover the original token order. 
    As a result, reconstruction methods attempting token-level fidelity face particular difficulty in restoring the precise temporal relationships between tokens, especially in longer sequences or architectures with complex components like stacked self-attention layers.

\end{itemize}

\textbf{\methodName.} 
To address the aforementioned challenges, we introduce \methodName, a more scalable and generalizable data reconstruction attack on LLMs that can effectively reconstruct training data with larger batch sizes and longer sequences, even under PEFT-induced gradient sparsity. Our analysis (Section~\ref{sec: rank}) reveals that neural network gradients, including those from PEFT methods, are inherently low-rank, constrained by the total number of tokens, and comprise linear combinations of the input embeddings. This implies that the input embeddings are confined within a low-dimensional subspace formed by these gradient vectors. By leveraging this subspace to guide the optimization search, we streamline the process, achieving more effective and efficient convergence. With this guided approach, \methodName~significantly enhances scalability by accommodating increases in batch size and sequence length. Notably, this \textit{gradient subspace guidance} also promotes better generalization and ensures true model-architecture-agnosticism. Since the gradients can be expressed as linear combinations of the input token embeddings weighted by their influence on the loss, regardless of whether the underlying model is encoder-based, decoder-based, or encoder-decoder, they inherently share a common linear structure. This consistent linear formulation across architectures allows \methodName to remain model-agnostic, guiding the optimization without relying on architecture-specific features like cross-attention, thereby controlling the combinatorial explosion and mitigating prohibitive computational costs.

Additionally, to specifically address PEFT-related sparsity, we employ \textit{null space regularization}. After computing the gradient updates, we construct a gradient matrix by stacking the gradients of the loss with respect to the input embeddings across the batch. 
This matrix captures the span of directions influenced by the loss. We identify directions in its null space that do not contribute to the loss and introduce a penalty term to constrain deviations along these directions.  This projection-based penalty ensures that each reconstructed embedding remains close to the true embedding by discouraging deviations in unconstrained dimensions. Consequently, the augmented objective balances consistency between constrained and unconstrained directions, mitigating noise and ambiguity that would otherwise arise from the sparse gradients typical of PEFT. Moreover, to improve the order correctness of the reconstructed token sequence, we introduce a \textit{Sequence Order Calibration} phase, in which we incorporate sequential token recovery aided by gradient alignment. Specifically, we reconstruct the correct token sequence by iteratively assigning tokens to positions based on gradient alignment. For each position, candidate tokens are evaluated by computing their isolated gradients and comparing these to the full sequence gradient, leveraging positional and contextual dependencies to determine the token with the highest alignment, which is then fixed before proceeding to subsequent positions.

We conducted extensive experiments on a range of LLMs with varying sizes and architectures, tailored to tasks such as sentiment analysis and next-word prediction. Our evaluation spans multiple datasets and PEFT methods in FL (i.e., SLoRA~\cite{babakniya2023slora} and FedAdapter~\cite{cai2022fedadapter}), along with seven state-of-the-art baseline attacks. The results consistently show that \methodName achieves enhanced effectiveness, superior scalability, and generalizability compared to existing attacks.

\section{Preliminaries \& Rank-Deficient Gradients}
\subsection{Data Reconstruction Attack}
Data reconstruction attack, a.k.a. gradient inversion attack, significantly challenges the privacy assurances of FL~\cite{jeon2021gradient}. In this attack, an adversary reconstructs a client's private data $(x_i, y_i)$ by leveraging the gradient updates $\nabla_{\theta} f(x_i, y_i)$ sent during a communication round. This scenario assumes the presence of an honest-but-curious server, which follows the FL protocol yet seeks to infer sensitive information. Additionally, the adversary could be a malicious analyst who eavesdrops on the communication channel. 
The objective function commonly used in such attacks is described as~\cite{geiping2020inverting}:
\begin{equation}\label{eq: matching}
    \underset{(x_i^*, y_i^*)}{\arg \min} \; \mathcal{L}_{rec}(\nabla_{\theta} f(x_i, y_i), \nabla_{\theta} f(x_i^*, y_i^*)),
\end{equation}
where $\mathcal{L}_{rec}$ denotes a distance loss and $*$ denotes the reconstructed or dummy data. These attacks can be adapted for LLM training with the next-word prediction by targeting the reconstruction of token sequences $\textbf{x} = [x_1, x_2, \dots, x_n]$, which lack discrete labels typical of classification tasks. The objective is to minimize the distance between the true gradients and the gradients of reconstructed sequences: 
\begin{equation}
    \underset{\mathbf{x}^{*}_{i}}{\arg \min} \; \mathcal{L}_{rec}\left(\nabla_{\theta} \sum_{t=1}^n \ell(x_t | x_{<t}), \; \nabla_{\theta} \sum_{t=1}^n \ell(x_t^* | x_{<t}^*)\right).
    \label{eq: next-token}
\end{equation}

\subsection{Transformers}
In this work, we address the problem of gradient inversion in transformer models~\cite{vaswani2017attention}, specifically focusing on LLMs used for text. The process begins with tokenizing the input text into tokens from a fixed vocabulary of size $\nu$. Each token is then converted into a one-hot vector, denoted as $\bm{t_1}, \bm{t_2}, \ldots, \bm{t_n} \in \mathbb{R}^\nu$, where $n$ is the sequence length. These tokens are subsequently transformed into embedding vectors $\bm{x_1}, \bm{x_2}, \ldots, \bm{x_n} \in \mathbb{R}^d$, where $d$ is the chosen embedding size, by multiplying them with the trained embedding weights $\bm{W_{embed}} \in \mathbb{R}^{\nu \times d}$.

In addition to token embeddings, token positions are encoded using positional embedding weights $\bm{W_{pos}} \in \mathbb{R}^{P \times d}$, where $P$ is the maximum allowed token sequence length. For an input sequence of length $n$ ($n\leq P$), positional embeddings $\bm{p_1}, \bm{p_2}, \ldots, \bm{p_n}$ are added to the corresponding token embeddings to form the input embeddings $\bm{z_1}, \bm{z_2}, \ldots, \bm{z_n}$. 
For a batch of $b$ sequences padded to a
common length $n$, stacking yields
$\mathbf{Z} \in \mathbb{R}^{b \times n \times d}$. The total number of tokens
processed in the batch is $t = b n$. 
%Stacking these embeddings for a batch of $b$ sequences yields the input matrix $\bm{Z} \in \mathbb{R}^{b \times n \times d}$, where $b$ is the batch size (number of sequences in one batch). Let $t$ denote the total number of tokens in the batch $t=\sum^{b}_{j=1} n_j$. 

%We use the flattened matrix $\tilde{Z} \in \mathbb{R}^{t\times d}$ formed by stacking all token embeddings; subsequent gradient expressions and rank bounds are stated in terms of $\tilde{Z}$.}

The stacked embedding $\bm{Z}$ is then passed through multiple layers of self-attention. We denote the input to the $k$-th self-attention layer as $\bm{Z_k} \in \mathbb{R}^{b \times n \times d}$ for $1 \leq k \leq L$, where $L$ is the number of transformer blocks. In each block, the self-attention mechanism involves computing three linear projections of the input embeddings to form the queries $\bm{Q} = \bm{Z_k W_k^Q}$, keys $\bm{K} = \bm{Z_k W_k^K}$, and values $\bm{V} = \bm{Z_k W_k^V}$. The attention scores are then computed as:
\begin{equation}
attention(\bm{Q}, \bm{K}, \bm{V}) = softmax\left(\bm{M} \odot \frac{\bm{QK^T}}{\sqrt{d}}\right)\bm{V},
\end{equation}
where $\bm{M}$ is the binary self-attention mask, $\odot$ denotes the element-wise multiplication, and the softmax operation is performed independently across each row.

\subsection{Parameter Efficient Fine-Tuning}
Due to the enormous size of LLMs, many recent studies have proposed using parameter-efficient fine-tuning (PEFT) techniques, such as Low-Rank Adaptation (LoRA)~\cite{hu2021lora} and Adapters~\cite{houlsby2019parameter}, to fine-tune these models for specific tasks without the need to retrain the entire network.  Specifically, LoRA introduces low-rank matrices into the transformer architecture, allowing only these low-rank components to be updated during fine-tuning, thereby reducing the number of parameters that need to be trained. The process begins with the pre-trained model weights, $\bm{W} \in \mathbb{R}^{d_{in} \times d_{out}}$, where $d_{in}$ is the input dimension and $d_{out}$ is the output dimension. LoRA introduces two low-rank matrices, $\bm{A} \in \mathbb{R}^{d_{in} \times r}$ and $\bm{B} \in \mathbb{R}^{r \times d_{out}}$, where $r \ll \min(d_{in}, d_{out})$, such that the weight update is approximated as $\Delta \bm{W} = \bm{AB}$. During fine-tuning, only the matrices $\bm{A}$ and $\bm{B}$ are updated, leaving the original model weights $\bm{W}$ unchanged. This approach significantly reduces the number of parameters that need to be optimized. However, in FL, where data distribution is highly non-IID, LoRA can experience slower convergence rates due to the initialization process of LoRA blocks~\cite{babakniya2023slora}. To address this, several variants such as SLoRA~\cite{babakniya2023slora}, FLoCoRA~\cite{ribeiro2024flocora}, and HetLoRA~\cite{cho2024heterogeneous} have been developed to better suit the FL environment, maintaining the foundational technique of LoRA.

Additionally, Adapters~\cite{houlsby2019parameter} offers a flexible and efficient approach for fine-tuning LLMs by inserting small, trainable modules into the transformer architecture. Adapters aim to reduce computational complexity while maintaining their performance by down-projecting the output of a specific layer, such as the MLP layer $\bm{H_o}$, to a lower dimension using a projection matrix $\bm{W_{\text{down}}}$. This is followed by applying a non-linear activation function and then up-projecting back to the original dimension using another projection matrix $\bm{W_{\text{up}}}$. Mathematically, $\bm{H_o}' = \bm{H_o} + f(\bm{H_o} \bm{W_{\text{down}}}) \bm{W_{\text{up}}},$ where $f(\cdot)$ represents a non-linear function such as ReLU. During fine-tuning, only the parameters of the adapter layers ($\bm{W_{\text{down}}}$ and $\bm{W_{\text{up}}}$) are updated, while the original model weights remain unchanged. However, adapters in FL face challenges due to the large configuration space, impacting training overhead and model convergence delays. To address these challenges, FL-specific adapters, such as FedAdapter~\cite{cai2022fedadapter}, and Feddat~\cite{chen2024feddat}, have been developed, enhancing adapters' efficiency and effectiveness in FL settings.

\begin{figure*}[th]
    \centering
    \includegraphics[width=0.67\linewidth]{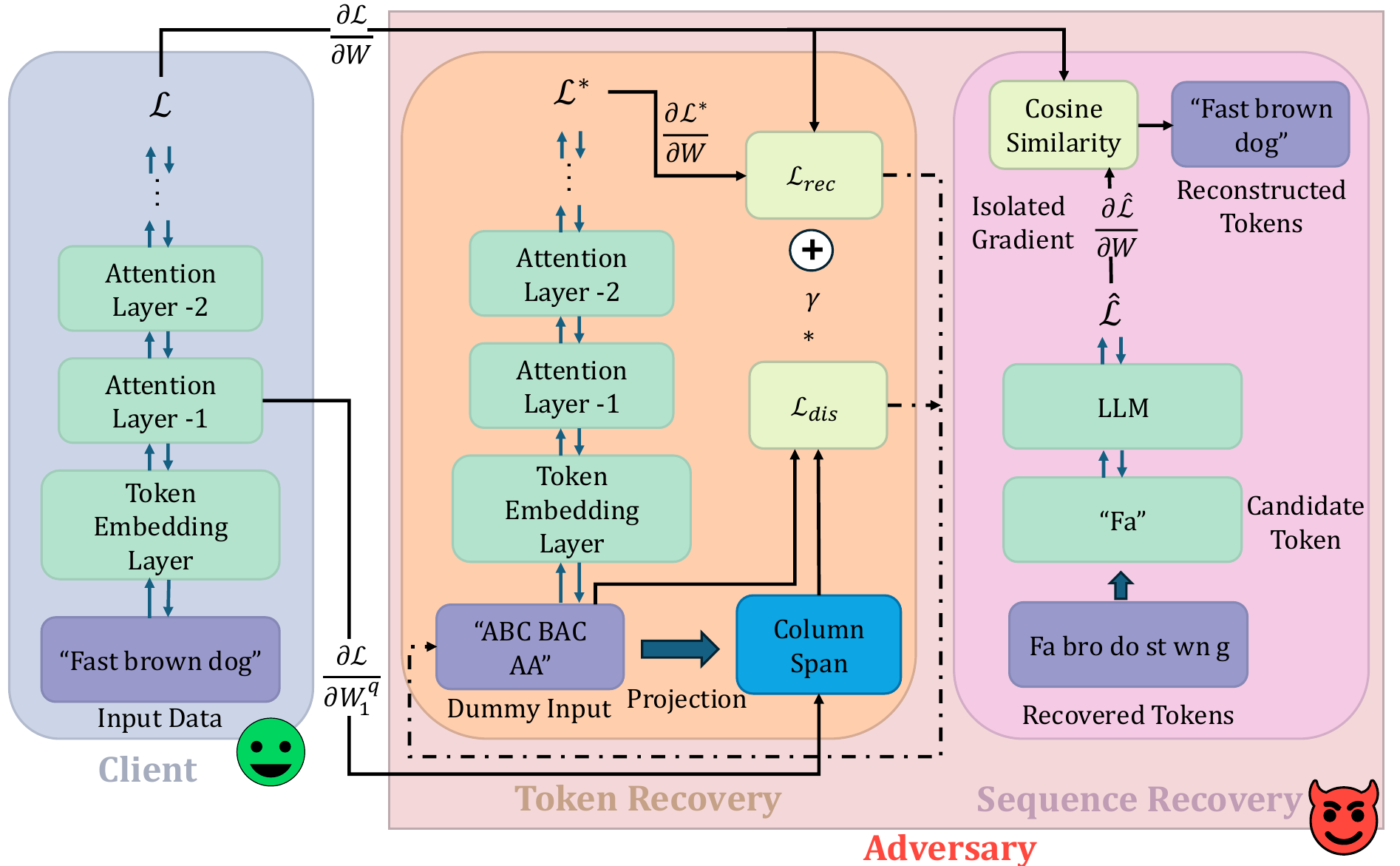}
    \caption{Overview of \methodName. \methodName enables an adversary to reconstruct client training data by initializing a dummy input and iteratively updating it to match the client’s gradient. To reduce the search space, the server projects candidate tokens onto the gradient’s column space and recovers the correct sequence by comparing individual token gradients iteratively.}
    \vspace{-2mm}
    \label{fig: overview}
\end{figure*}

\subsection{Gradient Rank Deficiency and Subspaces}\label{sec: rank}
\begin{theorem}\label{theorem:1}
As the attention layers of the transformer are linear layers (e.g., $\bm{Y} = \bm{ZW} + \bm{B}$, where $\bm{Y}$ represents the output, $\bm{Z}$ denotes the input, $\bm{W}$ is the layer's weight matrix, and $\bm{B}$ is the bias term), the gradients of the loss $\mathcal{L}$ w.r.t $\bm{W}$ for a linear layer can be expressed as:
\begin{equation}\label{eq: gradient}
    \frac{\partial \mathcal{L}}{\partial \bm{W}} = \bm{Z^T}\frac{\partial \mathcal{L}}{\partial \bm{Y}}.
\end{equation}
\end{theorem}
\begin{proof} Provided in Appendix~\ref{sec: theorem}. \end{proof}

It is known that the gradients of a neural network tend to be low-rank when the batch size is smaller than the hidden dimension~\cite{cai2017deep, jaderberg2014speeding}, and this observation holds for LLM architectures as well. Let $d$ denote the embedding dimension, $t$ denote the total number of tokens in a batch, $h$ the hidden state dimension, and $\mathcal{L}$ the loss function. The gradient matrix $\frac{\partial \mathcal{L}}{\partial \bm{Y}}$ has dimensions $t \times h$. Given that the input embedding matrix $\bm{Z}$ has dimensions $t \times d$, according to Theorem~\ref{theorem:1}, the gradient of the attention weight matrix (key, value, or query) is computed as the product of $\bm{Z^T}$ and $\frac{\partial \mathcal{L}}{\partial \bm{Y}}$. The rank of this product is bounded by the smallest dimension involved in the multiplication. Therefore, the rank of the gradient of the attention matrix is at most $\min(t, d, h)$.

\begin{theorem}\label{theorem:2}
Under the assumption that $\frac{\partial \bm{\mathcal{L}}}{\partial \bm{W}}$ is rank-deficient, $\bm{Z^T}$ is a linear combination of the columns of $\frac{\partial {\mathcal{L}}}{\partial \bm{W}}$.
\end{theorem}
\begin{proof} Provided in Appendix~\ref{sec: theorem}. \end{proof}
\vspace{-1mm}
In most training scenarios, the gradient matrix  $\frac{\partial \mathcal{L}}{\partial \bm{W}}$ is rank-deficient, with a rank at most $\min(t, d, h)$. 
According to Theorem~\ref{theorem:2}, the embedding vectors forming $\bm{Z^T}$ lie in the column span of $\frac{\partial \mathcal{L}}{\partial \bm{W}}$.
The adversary can exploit this to guide the optimization search by checking if the dummy/reconstructed embeddings lie within the span of the gradient.
This can be achieved by decomposing the matrix to separate it into simpler components, projecting the embedding vector onto the column space of the decomposed matrix, and computing the distance between the original vector and its projection. A small residual distance indicates the vector is within the span of the matrix.

It is important to note that PEFT gradients are typically low-rank, either by explicit design (e.g., LoRA) or due to the small size of the additional parameters (e.g., Adapters). Since the additional layers introduced by the PEFT methods are linear layers, Theorem~\ref{theorem:2} holds for PEFT methods as well. With PEFT methods, the effective rank is bounded by the PEFT matrix rank $r\ll d$, so rank $\frac{\partial \mathcal{L}}{\partial W}\leq min (t,d,h,r)$.

\section{Threat Model}
In this work, we consider an adversary that is either an honest-but-curious server or a malicious analyst capable of eavesdropping on the communication channel during the federated learning (FL) process. The adversary has access to both the global model shared by the server and the gradients uploaded by the clients during each communication round. We follow the standard white-box threat model adopted in prior gradient inversion and data reconstruction attacks~\cite{dengtaggradient, petrov2024dager, dimitrov2022lamp}, in which the adversary observes per-client gradients prior to aggregation; cryptographically secure aggregation mechanisms that hide individual client updates are not considered. The adversary aims to exploit this information to reconstruct the underlying training data and its contextual information.

\textbf{Adversarial Goal:} The adversary’s goal is to reconstruct private training data from gradients shared by clients. Unlike traditional gradient inversion attacks that rely on gradient leakage for token recovery, the adversary aims to recover richer contextual information from gradients to better approximate the training data while working within the constraints imposed by the FL process (e.g., frozen embedding layers).

\textbf{Adversary's Capabilities:} The adversary is assumed to have full access to the global model structure and parameters at each communication round. Additionally, the adversary can observe the gradients shared by clients but has no access to the clients’ local datasets or the aggregation mechanisms used by the server. The adversary is passive, it does not interfere with the training process, alter the global model distributed to clients, or tamper with the server-side aggregation. To further limit the adversary's capabilities, we adopt a standard FL fine-tuning practice by freezing the gradients of token and positional embeddings in federated LLMs~\cite{dimitrov2022lamp, petrov2024dager}. This mitigates the leakage of token-specific gradient updates that would otherwise simplify data reconstruction attacks. 

%Furthermore, we assume that if parameter-efficient fine-tuning (PEFT) is employed in the FL process, the adversary has knowledge of the specific PEFT method being used~\cite{tinn2023fine}. This is a reasonable assumption since PEFT strategies are executed on the client side, and knowledge of the method is necessary for local training, making it accessible to any adversarial client.
%When PEFT is used in the FL pipeline, we assume the adversary knows the presence of PEFT. This assumption aligns with prior federated PEFT works such as FedQLoRA~\cite{hu2025fedqlora}, FDLoRA~\cite{qi2024fdlora}, and FedP²EFT~\cite{lee2025fedp}, where either (i) the server defines and distributes the adapter structure to all clients, or (ii) clients are required to upload their adapter configurations before training begins. \textcolor{red}{In both cases, the PEFT configuration is shared across participants and is therefore observable to the adversary. Moreover, specific knowledge of PEFT configurations does not give the adversary an advantage.}
When PEFT is used in the FL pipeline, we assume the adversary knows \emph{whether} a PEFT method is adopted, but does not require knowledge of the specific method or its configuration. This assumption is reasonable because PEFT is executed client-side, and prior federated-PEFT studies (e.g., FedQLoRA~\cite{hu2025fedqlora}, FDLoRA~\cite{qi2024fdlora}, FedP$^{2}$EFT~\cite{lee2025fedp}) often (i) have the server define and distribute adapter structures to clients or (ii) ask clients to upload adapter configurations before training, practices that can make configurations visible to participants. Our threat model, however, does not rely on such configuration visibility; it only assumes adversarial awareness of PEFT adoption, not the method or its settings.

Our threat model considers federated LLM fine-tuning scenarios in which clients share gradient updates during training, regardless of whether the fine-tuning updates all model parameters or only a subset of parameters via PEFT methods. Accordingly, \methodName applies to both full-parameter fine-tuning and parameter-efficient fine-tuning, as both settings expose attention-layer gradients that may leak training data.

\section{Design of \methodName}
\subsection{Attack Overview}
\methodName is a scalable and generalizable data reconstruction attack designed for LLMs that effectively reconstructs training data under challenges such as larger batch sizes, longer sequences, and PEFT-induced gradient sparsity. Our approach leverages gradient subspace guidance to optimize the recovery process by utilizing the inherent low-rank structure of gradients. This method confines input embeddings to a low-dimensional subspace formed by gradient vectors, significantly improving convergence and enabling efficient reconstruction even with increased batch sizes and context lengths, addressing the challenge of scalability (for \circled{1}). 

To tackle the sparsity introduced by PEFT methods, \methodName employs null space regularization. This technique identifies directions in the gradient null space that do not influence the loss and applies a penalty to these directions, ensuring embeddings remain close to their true values. This approach mitigates the noise and ambiguity caused by sparse gradients, improving the reconstruction process in PEFT-based federated learning (for \circled{2}).

By focusing on the linear structure of gradients, which is consistent across encoder-based, decoder-based, and encoder-decoder architectures, \methodName achieves true architecture-agnosticism. It avoids reliance on network-specific features, such as cross-attention mechanisms, and projects optimization onto the most relevant gradient directions, thereby addressing the complexities posed by diverse model architectures (for \circled{3}). Furthermore, the attack employs gradient subspace guidance to streamline optimization, controlling combinatorial growth in the search space and mitigating excessive computational costs, effectively addressing time complexity concerns (for \circled{4}). To enhance the accuracy of token sequence recovery, \methodName introduces a Sequence Order Calibration phase. This phase iteratively aligns tokens by comparing isolated token gradients to full-sequence gradients, leveraging positional and contextual dependencies. By systematically fixing tokens in positions based on gradient alignment, this phase ensures a more accurate reconstruction of the original token sequence, addressing the challenge of sequence order (for \circled{5}).

Figure~\ref{fig: overview} illustrates the workflow of~\methodName, which comprises two main stages: Token Recovery and Sequence Recovery. In the Token Recovery phase, the adversary initializes a dummy input and iteratively updates it to minimize the gradient reconstruction loss ($\mathcal{L}_{rec}$), which computes the distance between the gradients of the dummy input and the actual gradients provided by the client (i.e., matching $\frac{\partial \mathcal{L}^*}{\partial \bm{W}}$ with $\frac{\partial \mathcal{L}}{\partial \bm{W}}$).
To reduce the high-dimensional search space, dummy token embeddings are projected onto the column span of the gradient matrix from the first attention layer ($\frac{\partial\mathcal {L}}{\partial \bm{W_1^q}}$), since this layer directly processes the input embeddings. This projection ensures that the recovered embeddings align with informative directions for reconstruction, guided by a distance loss ($\mathcal{L}_{dis}$). When a PEFT method is applied, the gradients become sparse, i.e., only a small subset of parameters receives non-zero updates, making it more difficult to identify informative directions. To address this, \methodName introduces a null space regularization term that penalizes updates in directions that do not affect the loss, thereby keeping the recovered embeddings within the gradient-contributing subspace.

In the Sequence Recovery phase, \methodName addresses the challenge that the recovered tokens from the Token Recovery phase are unordered and lack positional context. Although these tokens closely match the true embeddings, their order within the original sequence remains unknown. To recover the correct order of these tokens, \methodName employs a sequential order calibration strategy. Specifically, it evaluates each token candidate’s gradient alignment by computing the cosine similarity between its isolated gradient ($\frac{\partial \hat{\mathcal{L}}}{\partial \bm{W}}$), that is, the gradient when the token is placed in a specific position within a partially constructed sequence, and the original full-sequence gradient ($\frac{\partial \mathcal{L}}{\partial \bm{W}}$) received from the client. The token that yields the highest alignment score for each position in the sequence is selected and fixed at that position. This process is repeated iteratively for the remaining positions, progressively reconstructing the full sequence in the correct order while leveraging both positional and contextual dependencies embedded in the gradient.

\subsection{Token Recovery}
In \methodName, we aim to obtain the private data of a client using the shared gradient updates $\nabla_{\theta} f(x_i, y_i)$ on samples $(x_i, y_i)$ from their dataset $(\bm{X}, \bm{Y})$. We follow the gradient matching procedure by updating a set of dummy data $(x_i^*, y_i^*)$ to match the shared gradients. Thus, the optimization part of the attack solves Equation~\ref{eq: matching}. In the case of a classification task, the labels can be recovered for gradient and network architectures (details in section~\ref{sec:label}). For next-word prediction tasks, where labels are implicit, we apply the same objective by treating the token sequences as inputs and minimizing the gradient difference, as described in Equation~\ref{eq: next-token}.

\textbf{Reconstruction Loss.} Common choices for $\mathcal{L}_{rec}$ include $L_2$~\cite{zhao2020idlg}, $L_1$~\cite{dengtaggradient}, and cosine distances~\cite{dimitrov2022lamp}. However, relying solely on these metrics can be suboptimal: $L_1$ norm captures absolute differences but misses gradient direction; $L_2$ norm is sensitive to outliers and computationally expensive; cosine similarity ignores magnitude differences, missing full data variations. To address these issues, we propose a weighted layer-wise cosine similarity loss, which combines directional and magnitude information.
Specifically, we compute cosine similarity at each layer and weight it by the $L_1$ norm of that layer's gradient, allowing loss to emphasize more informative layers dynamically. The resulting loss is defined as:
\begin{align}
\mathcal{L}_{rec} = & \ 1 - \frac{1}{l} \sum_{j=1}^{l} \frac{\nabla_{\theta_{j}} f(x_i, y_i) \cdot \nabla_{\theta_{j}} f(x_i^*, y_i^*)}{\|\nabla_{\theta_{j}} f(x_i, y_i)\|_2 \|\nabla_{\theta_{j}} f(x_i^*, y_i^*)\|_2} \nonumber \\
& \qquad\qquad\qquad\qquad\quad \cdot \|\nabla_{\theta_{j}} f(x_i, y_i)\|_1,
\end{align}    
where $j$ is the layer index of the gradient and $l$ is the total number of layers. 

\textbf{Embedding Regularization.} As the batch size increases, it becomes challenging to extract data from the gradient due to the dilution of the gradient signal. To enhance the scalability of data reconstruction attacks when handling larger batch sizes and longer input sequences, we leverage the observation that gradients are rank-deficient when the batch size is smaller than the hidden dimension. 
According to Theorem~\ref{theorem:2}, in such cases, the input embeddings lie within the column span of the gradient matrix $\frac{\partial \mathcal{L}}{\partial \bm{W}}$. This insight allows us to assess whether the optimized dummy embedding $x_i^*$ lies in the same subspace. 
Since only the first transformer layer handles the input embedding, 
we focus on the gradients of its query matrix, $\frac{\partial \mathcal{L}}{\partial \bm{W^q_{1}}}$.
We project the dummy embedding $x_i^*$ onto the column span of this gradient. 
If $x_i^*$ lies entirely within the subspace spanned by the gradient matrix, the projection will be equal to $x_i^*$. This is because its component in the orthogonal complement of the subspace is zero, meaning it has no component in the orthogonal direction outside the subspace.
To enforce this constraint, we define a projection distance loss:
\begin{equation}
    \mathcal{L}_{dis} = ||\bm{Q}(\bm{Q^T} \cdot x_i^*) - x_i^*||_2,
\end{equation}
where $\bm{Q}$ is an orthogonal matrix obtained from the QR decomposition of the matrix $\frac{\partial \mathcal{L}}{\partial \bm{W^q_{1}}}$. We use the distance between $x_i^*$ and its projection $\bm{Q}(\bm{Q^T} \cdot x_i^*)$ as a regularization term in the optimization, penalizing $x_i^*$ that deviates from the column span of $\frac{\partial \mathcal{L}}{\partial \bm{W^q_{1}}}$. 

\textbf{Optimization-based Reconstruction.}
By incorporating the regularization term, the total loss function becomes:
\begin{equation}
\begin{split}
    \mathcal{L}_{\methodName} = \mathcal{L}_{rec} + \gamma \mathcal{L}_{dis} \quad v_{min} \leq x_i^* \leq v_{max}, \\(\text{element-wise}),
\end{split}
\end{equation}
where $\gamma$ is a weighting factor. The equation becomes purely gradient matching if $\gamma$ is $0$. During the optimization of the \methodName loss, we observe that the resulting embedding vectors $x_i^*$ often steadily grow in value. We believe this behavior results from the optimization process focusing on optimizing the direction of individual embeddings $x_i^*$ while disregarding their magnitude. To address this, we introduce a constraint that bounds the embedding values within a range $(v_{min}, v_{max})$, where these bounds are vectors containing the minimum and maximum values observed across each dimension of the embedding table. This ensures the reconstructed embeddings stay within a plausible range consistent with the original model’s vocabulary embeddings.

\subsection{Token Recovery under PEFT}
In parameter-efficient fine-tuning (PEFT) methods like LoRA and Adapters, the gradient matrix $\frac{\partial \mathcal{L}}{\partial \bm{W^q_1}}$ inherently has a low rank $r \ll d$, where $r$ is the rank of the matrix and $d$ is the embedding size. 
The column space of the gradient matrix, spanning $r$-dimensions, represents the constrained directions that influence the loss. In contrast, the null space $\mathcal{N} \subseteq \mathbb{R}^d$, which spans the remaining $d-r$ dimensions, contains directions that do not affect the gradient. Any embedding vector that differs from the “true” embedding $x$ in that null space direction will not be penalized by the first-order projection alone. The large null space introduced by low-rank constraints in PEFT methods complicates sequence reconstruction tasks, as the null space allows for many plausible sequences that produce similar gradients, creating significant ambiguity.

To solve this issue, we use a null space regularization method by providing additional constraints through projection-based penalties. Specifically, we compute the null space of the gradient matrix to identify directions that do not affect the loss and introduce a penalty term to constrain these directions. The null space projection ensures that the reconstructed embedding $x_i^*$ does not deviate significantly in unconstrained directions, reducing ambiguity and aligning it with the true embedding $x_i$. The reconstruction objective is augmented to include a null space penalty term, ensuring consistency across constrained and unconstrained directions. The updated reconstruction loss is defined as: 

\begin{equation}
\begin{split}
    \mathcal{L}_{\methodName} = \mathcal{L}_{rec} + \gamma \mathcal{L}_{dis} + \eta \|P_{\mathcal{N}} x_i^* \|^2,\\
    \text{s.t.} \; v_{min} \leq x_i^* \leq v_{max},
\end{split}
\end{equation}

where $P_{\mathcal{N}} = V_N V_N^T$ is the projection matrix onto the null space, and $V_N$ contains the basis vectors spanning the null space. These basis vectors are the columns of the orthogonal matrix formed from the right singular vectors corresponding to zero singular values ($\sigma_i = 0$). The orthogonal matrix is obtained by performing singular value decomposition (SVD) on the gradient matrix. 

%\textcolor{red}{Moreover, even if an adversary only knows that some form of PEFT has been applied, without access to the exact adapter configuration, \methodName can still execute the full inversion attack. The recovery strategy hinges on null-space regularization to constrain projection penalties, which is inherently agnostic to the adapter design. Consequently, the attack does not require PEFT-specific details and remains effective regardless of the particular fine-tuning variant adopted.}
%\jian{Do we have any results to support this claim?}

\subsection{Sequence Order Calibration}
Using the proposed optimization method, we can reconstruct the tokens present in the training data that are reflected in the gradients. However, since the gradients of the embedding layer and the positional encoding layer are frozen, the recovered tokens may not be in the correct sequence. For example, the original sequence might be ``\textit{The dog bit the man}," but the reconstruction could be ``\textit{The man bit the dog}." Although both sequences may yield similar perplexity losses due to the same words and grammatical correctness, their meanings differ significantly. To resolve this, we apply the intuition that the alignment between gradients for partial and full sequences plays a critical role in ensuring accurate recovery of the original training data. The gradient of a correctly placed partial sequence is more closely aligned with the full-sequence gradient than the gradient of an incorrectly placed partial sequence. Intuitively, this is because the gradients for the ``correct'' tokens contribute consistently to the overall objective, whereas misordered tokens introduce deviations that disrupt this alignment.
%while substituting the ``wrong'' tokens introduces deviations that disrupt this alignment.

\begin{lemma} \label{lemma:}
Let $\{\mathbf{g}_p\}_{p=1}^n$ be a sequence of vectors in $\mathbb{R}^d$. Let the total sum be $S \;=\; \sum_{p=1}^n \mathbf{g}_p$. For any \( k < n \), define the partial sum $S_{\le k} \;=\; \sum_{p=1}^k \mathbf{g}_p$. 
Then,
\[
S_{\le k} \cdot S
\;=\;
\bigl\|S_{\le k}\bigr\|^2
\;+\;
S_{\le k} \,\cdot\, \sum_{p=k+1}^n \mathbf{g}_p.
\]
In particular, if the remaining vectors \( \{\mathbf{g}_p\}_{p=k+1}^n \) are not strongly negatively correlated with \( S_{\le k} \), then the dot product \( S_{\le k} \cdot S \) will remain large.
\end{lemma}
\begin{proof} Provided in Appendix~\ref{sec: theorem}. \end{proof}
The key intuition is that, in next-word prediction tasks, gradients are computed over sequentially structured inputs where each token depends on previous ones. As a result, individual gradients $ g_1, g_2, \dots, g_n $ (e.g., from each token position) are not arbitrary, reflecting shared contextual structure. For example, gradients from earlier tokens encode valid prefixes of the sequence. Therefore, the partial sum $ S_{\le k} = \sum_{p=1}^k g_p $ not only contributes to the total gradient $ S = \sum_{p=1}^n g_p $, but also tends to point in a similar direction. This alignment arises from partial and full sums composed of gradients shaped by overlapping context, resulting in directional consistency across sequence. Based on this, we formalize the following theorem.

\begin{theorem}\label{theorem}
Consider a full sequence $\textbf{x} = (x_1, x_2, \dots, x_n)$ and two partial sequences of length $k$: A correctly ordered partial sequence $x_{\le k}^{\text{(correct)}}, e.g., (x_1, x_2, \dots, x_k)$, and a misordered partial sequence $x_{\le k}^{\text{(wrong)}}$, e.g., $(x_1, x_3, \dots)$. Let 

\begin{equation}
    \nabla_\theta^{\text{full}}:==-\sum_{t=1}^n \nabla_\theta \log p_\theta\bigl(x_t \mid x_{<t}\bigr),
\end{equation}
\begin{equation}
    \nabla_\theta^{\text{correct}} := -\sum_{t=1}^k \nabla_\theta \log p_\theta\bigl(x_t \mid x_{<t}\bigr),
\end{equation}

\begin{equation}
\begin{split}
\nabla_\theta^{\text{wrong}}:= -\Bigl[\nabla_\theta \log p_\theta(x_1) + \nabla_\theta \log p_\theta(x_3 \mid x_1) + \dots \\
+ \nabla_\theta \log p_\theta(x_t \mid x_{t-1})\Bigr],
\end{split}
\end{equation}

Then,
\[ \nabla_\theta^{\text{correct}} \cdot \nabla_\theta^{\text{full}} \;>\; \nabla_\theta^{\text{wrong}} \cdot \nabla_\theta^{\text{correct}},
\]
indicating that the gradient for the correctly ordered partial sequence is more aligned with the gradient of the full sequence than the gradient for the misordered partial sequence.
\end{theorem}
\begin{proof} Provided in Appendix~\ref{sec: theorem}. \end{proof}

Applying Theorem~\ref{theorem}, we leverage the alignment between token gradients and the observed full-sequence gradient to refine token order. Specifically, the model’s training gradient for the entire sequence serves as the reference to guide reordering. Note that if the model’s original task is classification, the adversary must first adapt it to a next-word prediction task to expose token-level gradients.

With the recovered token set $(x^*_1, x^*_2, x^*_3, \dots, x^*_n)$ from the Token Recovery stage, the adversary begins with an empty sequence and iteratively assigns tokens to positions. For the first position, the adversary places each token $x_i^*$ from the set in position 1, computes the gradient of the model’s loss with respect to its parameters, and compares the similarity with the observed gradient using cosine similarity. The token whose gradient aligns most closely with the observed gradient is selected for the position. This process is repeated for subsequent positions. At each step, the adversary evaluates the remaining tokens by placing each candidate in the current position, computing the gradient, and comparing it to the observed gradient. The token with the highest alignment is fixed in the position, and the sequence is incrementally constructed until all tokens are placed. This method leverages the alignment between token-specific gradients and the full-sequence gradient, ensuring that the reconstructed sequence captures the positional and contextual dependencies encoded in the model during training.

\subsection{Label Recovery for Classification Tasks}
\label{sec:label}
To recover the label in a classification task, we exploit the gradients of the classification layer, which maps the LLM's hidden representations to the binary output space. Specifically, we analyze the structure of shared gradients and their relationship to the hidden representations to infer the correct label.

For a batch of tokenized text sequences $\mathbf{X} = \{x_1, x_2, \ldots, x_b\}$ and their corresponding one-hot labels $\mathbf{Y} = \{y_1, y_2, \ldots, y_b\}$, where $y_k \in \{0, 1\}^2$ and batch size is $b$. For each input $x_k$, the binary classifier produces logits $\mathbf{z}_k \in \mathbb{R}^2$, which are subsequently transformed into class probabilities $\mathbf{p}_k = \text{softmax}(\mathbf{z}_k)$. The cross-entropy loss for an individual instance is defined as $\mathcal{L}(x_k, y_k) = - \sum_{c=1}^2 y_{k,c} \log(p_{k,c})$, where $p_{k,c}$ represents the predicted probability for class $c$. The gradient of the loss with respect to the logits, $\partial \mathcal{L}(x_k, y_k) / \partial z_{k,c}$, measures the discrepancy between the predicted and ground-truth probabilities and is computed as $p_{k,c} - y_{k,c}$.

The final fully connected layer of the LLM-based classifier maps the hidden representation $\mathbf{h}_k \in \mathbb{R}^M$ of input $x_k$ to logits $\mathbf{z}_k$ through a weight matrix $\mathbf{W}^{(FC)} \in \mathbb{R}^{M \times 2}$. The gradient of the loss with respect to $\mathbf{W}^{(FC)}$ is given by $\Delta \mathbf{W}^{(FC)}_{m,c,k} = (p_{k,c} - y_{k,c}) \cdot h_{k,m},$ where $h_{k,m}$ denotes the $m$-th feature of $\mathbf{h}_k$. To analyze the gradients at the batch level, the average gradient over all inputs in the batch are expressed as $\Delta \mathbf{W}^{(FC)}_{m,c} = \frac{1}{b} \sum_{k=1}^b \Delta \mathbf{W}^{(FC)}_{m,c,k}$.

To infer the labels, we aggregate the gradients over the feature dimension $m$, resulting in the per-sample contribution for each class $c$, $S_{k,c} = \sum_{m=1}^M \Delta \mathbf{W}^{(FC)}_{m,c,k} = (p_{k,c} - y_{k,c}) \cdot \|\mathbf{h}_k\|_1$, where $\|\mathbf{h}_k\|_1$ is the L1 norm of the hidden representation. The sign of $S_{k,c}$ indicates whether the gradient contribution supports or contradicts a given label, being non-positive for the true label and non-negative otherwise. Using this property, we estimate the likelihood of each label $c$ by evaluating the aggregated gradient contributions:
\begin{equation}
\sum_{k=1}^b y_{k,c} \approx \sum_{k=1}^b p_{k,c} - 
\frac{b \cdot S_c}{\overline{\|\mathbf{h}_k\|_1}},
\end{equation}
where $S_c = \frac{1}{b}\sum_{k=1}^b S_{k,c}$ is the average contribution of gradients for class $c$, and $\overline{\|\mathbf{h}_k\|_1} = \frac{1}{b}\sum_{k=1}^b \|\mathbf{h}_k\|_1$ is the average L1 norm of the hidden representations across the batch. This approach circumvents the need for explicit reconstruction of the input data and directly infers the label distribution. The binary nature of the classification task simplifies the inference process by reducing ambiguity in label estimation.

However, for next-word prediction tasks, generating explicit ground truth labels is inherently challenging because the input text sequence itself serves as the ground truth. Unlike traditional classification tasks, where labels are predefined and discrete, next-word prediction involves predicting a token in a continuous sequence, where the target output depends on the preceding context. Consequently, in this scenario, explicit label inference is not directly required. Instead, the subsequent token in the dummy input sequence is treated as the ground truth for gradient-based analysis or restoration.

\section{Evaluation}
\subsection{Experimental Setup}
In this work, we evaluate \methodName on Large Language Models (LLMs) encompassing both encoder and decoder transformer architectures.
%, using the FedSGD aggregation method. 
To ensure diversity in architectural designs and enable comprehensive comparisons with various baselines, we select four widely used LLMs: GPT-2~\cite{radford2019language}, BERT$_{Base}$\cite{devlin2018bert}, Llama-7B\cite{touvron2023llama}, and T5$_{base}$~\cite{raffel2020exploring}. These models represent all three primary transformer architectures—decoder-based, encoder-based, and encoder-decoder—allowing us to analyze performance across distinct design paradigms. Additionally, for testing on larger models, we extend our evaluation to Llama-3.1 across multiple sizes.

Our experiments focus on two key tasks commonly addressed by transformer models: binary classification and next-word prediction. For binary classification tasks like sentiment analysis, we use the CoLA dataset\cite{warstadt2019neural}, which consists of English sentences labeled for grammatical acceptability, and the Rotten Tomatoes dataset\cite{socher2013recursive}, which contains movie reviews annotated for positive or negative sentiment. For next-word predictions, we employ the MIMIC-III dataset~\cite{johnson2016mimic}, a large-scale clinical dataset containing electronic health records, where the task involves predicting the next token in a sequence of clinical notes. Each attack method was executed $15$ times, with results averaged across recommended iterations per run to reduce variability and ensure robust evaluation of \methodName across diverse models and tasks.

In our experiments, each dataset is randomly partitioned across $100$ clients. The server may choose to attack one or more of these clients, and \methodName is executed separately for each targeted client. For every run, we randomly sample an iteration at which to launch the attack and report results averaged over 10 runs. Within each run, reconstruction is performed batch-by-batch over the training data, and performance is computed per attacked batch before aggregating across runs. The hyperparameters $\gamma$ and $\eta$ were tuned via grid search. The details of the baseline implementation are provided in Appendix~\ref{app: baseline}.

\begin{table*}[t]
\centering
\caption{Comparison of sequence reconstruction from gradients between~\methodName and other baseline algorithms on various batch sizes and datasets with BERT$_{Base}$ target model. R-1, R-2, and R-L denote the ROUGE-1, ROUGE-2, and ROUGE-L scores, respectively. All results are reported as mean $\pm$ standard deviation over 10 independent runs.}
 \resizebox{\textwidth}{!}{
\begin{tabular}{ccccccccccccccccc}
\hline
     \multirow{2}{4em}{\textbf{Dataset}}&\multirow{2}{4em}{\textbf{Method}} &\multicolumn{3}{c}{$b$ = 1} & &\multicolumn{3}{c}{$b$ = 2} &&\multicolumn{3}{c}{$b$ = 4} &&\multicolumn{3}{c}{$b$ = 8}\\\cline{3-5} \cline{7-9} \cline{11-13} \cline{15-17}
     
     & &R-1 &R-2 &R-L &&R-1 &R-2 &R-L &&R-1 &R-2 &R-L &&R-1 &R-2 &R-L\\
     \hline
     \multirow{7}{4em}{CoLA}
& DLG &58.6$_{\pm1.2}$ &7.6$_{\pm0.4}$ &45.5$_{\pm1.1}$ &&36.2$_{\pm1.3}$ &2.5$_{\pm0.3}$ &30.9$_{\pm1.2}$ &&34.8$_{\pm1.5}$ &1.3$_{\pm0.2}$ &31.4$_{\pm1.4}$ &&16.1$_{\pm1.0}$ &0.7$_{\pm0.1}$ &7.6$_{\pm0.9}$ \\
& TAG &78.1$_{\pm1.0}$ &10.0$_{\pm0.5}$ &52.7$_{\pm1.1}$ &&44.9$_{\pm1.2}$ &4.5$_{\pm0.4}$ &36.2$_{\pm1.0}$ &&34.9$_{\pm1.3}$ &1.5$_{\pm0.2}$ &30.8$_{\pm1.1}$ &&32.8$_{\pm1.2}$ &1.5$_{\pm0.2}$ &29.9$_{\pm1.0}$ \\
& LAMP &83.9$_{\pm0.9}$ &45.5$_{\pm1.1}$ &72.4$_{\pm1.0}$ &&56.4$_{\pm1.0}$ &21.4$_{\pm0.9}$ &49.1$_{\pm0.9}$ &&39.8$_{\pm1.1}$ &6.2$_{\pm0.5}$ &35.6$_{\pm1.0}$ &&35.8$_{\pm1.1}$ &4.9$_{\pm0.4}$ &33.9$_{\pm0.9}$ \\
& APRIL &83.8$_{\pm0.8}$ &45.3$_{\pm1.0}$ &72.1$_{\pm0.9}$ &&56.2$_{\pm0.9}$ &21.5$_{\pm0.8}$ &48.9$_{\pm0.8}$ &&40.6$_{\pm1.0}$ &7.6$_{\pm0.6}$ &39.5$_{\pm1.0}$ &&36.6$_{\pm1.0}$ &5.0$_{\pm0.4}$ &33.8$_{\pm0.8}$ \\
& FILM &84.0$_{\pm0.9}$ &44.6$_{\pm1.0}$ &71.8$_{\pm1.0}$ &&56.6$_{\pm1.0}$ &18.9$_{\pm0.8}$ &49.2$_{\pm0.9}$ &&43.2$_{\pm1.1}$ &11.1$_{\pm0.7}$ &39.6$_{\pm1.1}$ &&37.1$_{\pm1.0}$ &5.7$_{\pm0.5}$ &34.2$_{\pm0.9}$ \\
& BGP &85.8$_{\pm0.7}$ &50.7$_{\pm1.0}$ &75.9$_{\pm0.8}$ &&68.7$_{\pm0.8}$ &30.6$_{\pm0.9}$ &59.9$_{\pm0.8}$ &&49.8$_{\pm1.0}$ &11.5$_{\pm0.6}$ &43.2$_{\pm0.9}$ &&40.1$_{\pm0.9}$ &8.0$_{\pm0.6}$ &37.6$_{\pm0.8}$ \\
& DAGER &\textbf{99.6$_{\pm0.2}$} &\textbf{99.3$_{\pm0.3}$} &\textbf{99.1$_{\pm0.3}$} &&\textbf{97.8$_{\pm0.4}$} &\textbf{97.1$_{\pm0.5}$} &\textbf{93.6$_{\pm0.6}$} &&89.6$_{\pm0.9}$ &84.9$_{\pm1.0}$ &81.5$_{\pm1.1}$ &&63.6$_{\pm1.3}$ &44.3$_{\pm1.4}$ &43.9$_{\pm1.3}$ \\
& \textbf{\methodName} &94.5$_{\pm0.5}$ &93.9$_{\pm0.6}$ &92.8$_{\pm0.6}$ &&93.7$_{\pm0.6}$ &92.1$_{\pm0.7}$ &90.1$_{\pm0.7}$ &&\textbf{91.9$_{\pm0.7}$} &\textbf{90.4$_{\pm0.8}$} &\textbf{87.2$_{\pm0.8}$} &&\textbf{79.8$_{\pm0.9}$} &\textbf{71.9$_{\pm1.0}$} &\textbf{64.7$_{\pm1.0}$} \\
\hline
    
    \multirow{7}{4em}{\centering Rotten Tomatoes}
& DLG &19.7$_{\pm1.6}$ &0.4$_{\pm0.3}$ &14.9$_{\pm1.5}$ &&18.5$_{\pm1.7}$ &0.6$_{\pm0.3}$ &15.0$_{\pm1.6}$ &&18.3$_{\pm1.8}$ &0.4$_{\pm0.3}$ &15.2$_{\pm1.7}$ &&19.6$_{\pm1.9}$ &0.3$_{\pm0.2}$ &16.5$_{\pm1.8}$ \\
& TAG &31.2$_{\pm1.5}$ &2.4$_{\pm0.5}$ &19.7$_{\pm1.4}$ &&26.4$_{\pm1.6}$ &1.0$_{\pm0.4}$ &18.7$_{\pm1.5}$ &&27.4$_{\pm1.7}$ &0.8$_{\pm0.4}$ &19.7$_{\pm1.6}$ &&22.1$_{\pm1.8}$ &0.7$_{\pm0.3}$ &18.1$_{\pm1.6}$ \\
& LAMP &62.7$_{\pm1.3}$ &13.4$_{\pm1.1}$ &42.0$_{\pm1.2}$ &&37.8$_{\pm1.4}$ &6.1$_{\pm0.8}$ &28.2$_{\pm1.3}$ &&24.1$_{\pm1.5}$ &2.2$_{\pm0.6}$ &19.6$_{\pm1.4}$ &&20.2$_{\pm1.6}$ &0.6$_{\pm0.4}$ &17.3$_{\pm1.4}$ \\
& APRIL &63.5$_{\pm1.3}$ &15.1$_{\pm1.2}$ &43.2$_{\pm1.2}$ &&38.2$_{\pm1.4}$ &5.6$_{\pm0.8}$ &28.4$_{\pm1.3}$ &&27.7$_{\pm1.5}$ &2.3$_{\pm0.6}$ &20.7$_{\pm1.4}$ &&21.9$_{\pm1.6}$ &1.0$_{\pm0.5}$ &18.4$_{\pm1.4}$ \\
& FILM &63.6$_{\pm1.4}$ &15.3$_{\pm1.2}$ &43.6$_{\pm1.3}$ &&39.6$_{\pm1.5}$ &5.2$_{\pm0.9}$ &28.3$_{\pm1.4}$ &&30.5$_{\pm1.6}$ &2.5$_{\pm0.7}$ &23.1$_{\pm1.5}$ &&22.3$_{\pm1.7}$ &1.2$_{\pm0.5}$ &18.3$_{\pm1.5}$ \\
& BGP &71.5$_{\pm1.2}$ &20.4$_{\pm1.3}$ &48.7$_{\pm1.2}$ &&43.9$_{\pm1.3}$ &6.7$_{\pm0.9}$ &31.2$_{\pm1.3}$ &&29.3$_{\pm1.4}$ &3.3$_{\pm0.8}$ &23.8$_{\pm1.3}$ &&23.1$_{\pm1.5}$ &1.6$_{\pm0.6}$ &19.3$_{\pm1.4}$ \\
& DAGER &\textbf{99.4$_{\pm0.4}$} &\textbf{99.1$_{\pm0.5}$} &\textbf{98.8$_{\pm0.5}$} &&\textbf{94.7$_{\pm0.7}$} &\textbf{92.2$_{\pm0.8}$} &\textbf{85.4$_{\pm0.9}$} &&63.6$_{\pm1.4}$ &42.1$_{\pm1.5}$ &30.7$_{\pm1.5}$ &&32.3$_{\pm1.6}$ &6.8$_{\pm0.9}$ &6.8$_{\pm0.9}$ \\
& \textbf{\methodName}
&91.6$_{\pm0.8}$ &82.6$_{\pm0.9}$ &81.9$_{\pm0.9}$ &&89.8$_{\pm0.9}$ &78.6$_{\pm1.0}$ &63.1$_{\pm1.0}$ &&\textbf{77.2$_{\pm1.0}$} &\textbf{63.6$_{\pm1.1}$} &\textbf{59.5$_{\pm1.1}$} &&\textbf{61.9$_{\pm1.1}$} &\textbf{41.6$_{\pm1.2}$} &\textbf{57.6$_{\pm1.2}$} \\
\hline

    \multirow{7}{4em}{\centering MIMIC-III}
& DLG &13.4$_{\pm1.9}$ &9.4$_{\pm1.6}$ &27.6$_{\pm1.8}$ &&8.8$_{\pm2.0}$ &3.4$_{\pm1.4}$ &6.9$_{\pm1.6}$ &&5.4$_{\pm2.1}$ &0.5$_{\pm0.6}$ &1.3$_{\pm1.2}$ &&4.3$_{\pm2.2}$ &0.4$_{\pm0.6}$ &1.4$_{\pm1.2}$ \\
& TAG &13.9$_{\pm1.8}$ &10.0$_{\pm1.6}$ &10.9$_{\pm1.6}$ &&9.7$_{\pm1.9}$ &5.3$_{\pm1.5}$ &7.7$_{\pm1.5}$ &&6.1$_{\pm2.0}$ &0.7$_{\pm0.7}$ &0.7$_{\pm1.3}$ &&5.0$_{\pm2.1}$ &0.4$_{\pm0.6}$ &1.6$_{\pm1.3}$ \\
& LAMP &30.8$_{\pm1.6}$ &10.0$_{\pm1.5}$ &10.8$_{\pm1.6}$ &&10.9$_{\pm1.7}$ &1.2$_{\pm0.8}$ &4.3$_{\pm1.4}$ &&7.1$_{\pm1.8}$ &1.4$_{\pm0.9}$ &2.0$_{\pm1.4}$ &&6.8$_{\pm1.9}$ &0.6$_{\pm0.7}$ &2.3$_{\pm1.4}$ \\
& APRIL &32.1$_{\pm1.6}$ &8.9$_{\pm1.4}$ &12.1$_{\pm1.6}$ &&14.3$_{\pm1.7}$ &2.1$_{\pm0.9}$ &4.7$_{\pm1.4}$ &&2.9$_{\pm1.8}$ &1.0$_{\pm0.8}$ &1.7$_{\pm1.3}$ &&2.0$_{\pm1.9}$ &0.2$_{\pm0.5}$ &1.3$_{\pm1.2}$ \\
& FILM &30.6$_{\pm1.7}$ &7.0$_{\pm1.4}$ &12.8$_{\pm1.6}$ &&9.9$_{\pm1.8}$ &1.4$_{\pm0.9}$ &3.7$_{\pm1.4}$ &&3.0$_{\pm1.9}$ &1.1$_{\pm0.8}$ &1.8$_{\pm1.3}$ &&2.9$_{\pm2.0}$ &0.6$_{\pm0.7}$ &1.4$_{\pm1.3}$ \\
& BGP &42.6$_{\pm1.5}$ &22.9$_{\pm1.6}$ &27.4$_{\pm1.6}$ &&36.1$_{\pm1.6}$ &19.6$_{\pm1.7}$ &23.4$_{\pm1.6}$ &&14.7$_{\pm1.8}$ &7.4$_{\pm1.2}$ &8.8$_{\pm1.4}$ &&14.0$_{\pm1.9}$ &6.9$_{\pm1.2}$ &7.8$_{\pm1.4}$ \\
& DAGER &85.5$_{\pm0.9}$ &42.2$_{\pm1.1}$ &53.9$_{\pm1.2}$ &&80.5$_{\pm1.1}$ &42.8$_{\pm1.3}$ &50.6$_{\pm1.3}$ &&41.9$_{\pm1.5}$ &31.3$_{\pm1.6}$ &22.7$_{\pm1.5}$ &&22.4$_{\pm1.8}$ &0.0$_{\pm0.0}$ &0.0$_{\pm0.0}$ \\
& \textbf{\methodName}
&\textbf{86.6$_{\pm0.9}$} &\textbf{56.7$_{\pm1.1}$} &\textbf{72.9$_{\pm1.1}$} &&\textbf{81.8$_{\pm1.0}$} &\textbf{54.6$_{\pm1.2}$} &\textbf{62.2$_{\pm1.2}$} &&\textbf{66.2$_{\pm1.2}$} &\textbf{52.6$_{\pm1.3}$} &\textbf{58.6$_{\pm1.3}$} &&\textbf{50.9$_{\pm1.4}$} &\textbf{30.6$_{\pm1.5}$} &\textbf{50.7$_{\pm1.4}$} \\
\hline
\end{tabular}
}
\vspace{-2mm}
\label{tab: baselines}
\end{table*}

\subsection{Evaluation Metrics}
Following previous works~\cite{dimitrov2022lamp, dengtaggradient}, we evaluate our attack performance using the ROUGE metric suite~\cite{lin2004rouge}. Specifically, we report the F-scores for ROUGE-1, ROUGE-2, and ROUGE-L. ROUGE-1/2 measures the overlap of unigrams (individual words) and bigrams (pairs of consecutive words) between the generated sequence and the reference sequence. This metric provides a straightforward assessment of word-level retrieval accuracy, capturing contextual information to indicate how well the reconstructed sequence captures the original words. ROUGE-L assesses the Longest Common Subsequence (LCS) between the generated and reference sequences. It considers the order of tokens and measures the proportion of the longest continuous matching subsequence in relation to the entire reference sequence, emphasizing the preservation of sequence structure. In addition to ROUGE, we report entity-level F1, computed using a standard Named Entity Recognition (NER) model~\cite{wang2019cross}. This metric quantifies the recovery of privacy-sensitive entities, such as person names and clinical concepts, which may not be fully captured by surface-level lexical overlap alone.

\subsection{PEFT Methods}
In our evaluation, we use SLoRA~\cite{babakniya2023slora} and FedAdapter~\cite{cai2022fedadapter}, both of which are adaptations of LoRA and Adapter techniques. SLoRA integrates standard FL training with matrix decomposition to achieve a favorable initialization, addressing the slower convergence rates typically seen with standard LoRA in FL. 
FedAdapter progressively upgrades the adapter configuration throughout a training session to quickly learn shallow knowledge and incrementally learn deep knowledge by incorporating deeper and larger adapters.

%Following previous work, we assume the lengths of sequences are known for both baselines and our attacks, as an adversary can run the attack for all possible lengths [Balunovic et al., 2022].
\begin{table*}[t]
\centering
\caption{Comparison of sequence reconstruction from gradients between~\methodName and DAGER on larger batch sizes with different datasets and various model architectures. R-1, R-2, and R-L denote the ROUGE-1, ROUGE-2, ROUGE-L scores respectively.}
 \resizebox{\textwidth}{!}{
\begin{tabular}{cccccccccccccccccc}
\hline
     \multirow{2}{4em}{\textbf{Model}}&\multirow{2}{4em}{\textbf{Dataset}}&\multirow{2}{4em}{\textbf{Method}} &\multicolumn{3}{c}{$b$ = 16} & &\multicolumn{3}{c}{$b$ = 32} &&\multicolumn{3}{c}{$b$ = 64} &&\multicolumn{3}{c}{$b$ = 128}\\\cline{4-6} \cline{8-10} \cline{12-14} \cline{16-18}
     
     && &R-1 &R-2 &R-L &&R-1 &R-2 &R-L &&R-1 &R-2 &R-L &&R-1 &R-2 &R-L\\
     \hline
     \multirow{6}{4em}{\centering GPT-2}&\multirow{2}{4em}{\centering CoLA} &DAGER  &\textbf{100} &\textbf{100} &\textbf{63.7} &&\textbf{91.2} &\textbf{87.6} &\textbf{56.1} &&\textbf{86.3} &\textbf{71.6} &\textbf{43.2} &&24.2 &11.3 &10.5\\
     
     && {\methodName}  &{93.2} & {76.6} & {63.7} && {63.6} & {42.1} & {29.5} && {47.7} & {42.8} & {6.9} && \textbf{33.3} & \textbf{28.5} & \textbf{15.9}\\\cline{2-18}

     &\multirow{2}{4em}{\centering Rotten Tomatoes} &DAGER  &\textbf{92.2} &{65.1} &{44.1} &&\textbf{75.3} &{34.1} &{32.2} &&0.0 &0.0 &0.0 &&0.0 &0.0 &0.0\\
     
     && {\methodName}  &87.0 & \textbf{73.4} & \textbf{50.6} && {65.8} & \textbf{61.3} & \textbf{28.5} && \textbf{36.5} & \textbf{29.1} & \textbf{19.4} && \textbf{27.8} & \textbf{11.9} & \textbf{3.2}\\\cline{2-18}

     &\multirow{2}{4em}{\centering MIMIC-III} &DAGER  &{71.2} &{34.1} &{22.2} &&{41.2} &{10.1} &{5.2} &&1.9 &0.0 &0.0 &&0.0 &0.0 &0.0\\
     
     && {\methodName}  &\textbf{77.0} & \textbf{70.4} & \textbf{56.2} && \textbf{53.6} & \textbf{39.4} & \textbf{34.7} && \textbf{18.4} & \textbf{16.0} & \textbf{11.0} && \textbf{20.2} & \textbf{7.6} & \textbf{3.1}\\
     \hline

     \multirow{6}{4em}{\centering BERT$_{Base}$}&\multirow{2}{4em}{\centering CoLA} &DAGER  &{40.9} &{22.3} &{21.6} &&{15.3} &{4.2} &{3.8} &&2.1 &0.0 &0.0 &&0.0 &0.0 &0.0\\
     
     && {\methodName}  &\textbf{66.6} &\textbf{54.7} &\textbf{45.5} &&\textbf{45.4} &\textbf{30.1} &\textbf{21.1} &&\textbf{34.1} &\textbf{30.6} &\textbf{4.9} &&\textbf{23.8} &\textbf{23.2} &\textbf{4.2}\\\cline{2-18}

     &\multirow{2}{4em}{\centering Rotten Tomatoes} &DAGER  &{9.7} &{1.3} &{1.2} &&{5.9} &{0.0} &{0.0} &&4.2 &0.0 &0.0 &&0.0 &0.0 &0.0\\
     
     && {\methodName}  &\textbf{54.4} &\textbf{45.9} &\textbf{31.6} &&\textbf{41.1} &\textbf{38.3} &\textbf{17.8} &&\textbf{22.8} &\textbf{18.2} &\textbf{12.1} &&\textbf{18.5} &\textbf{7.9} &\textbf{2.1}\\\cline{2-18}

     &\multirow{2}{4em}{\centering MIMIC-III} &DAGER  &{5.1} &{1.7} &{1.0} &&{2.7} &{0.0} &{0.0} &&2.7 &0.0 &0.0 &&0.0 &0.0 &0.0\\
     
     && {\methodName}  &\textbf{42.8} &\textbf{39.1} &\textbf{31.2} &&\textbf{29.8} &\textbf{21.9} &\textbf{19.3} &&\textbf{10.2} &\textbf{8.9} &\textbf{6.1} &&\textbf{11.2} &\textbf{4.2} &\textbf{1.7}\\
     \hline

     \multirow{6}{4em}{\centering T5$_{base}$}&\multirow{2}{4em}{\centering CoLA} &DAGER  &{45.0} & {24.5} & {23.8} && {16.8} & {4.6} & {4.2} && {2.3} & {0.0} & {0.0} && {0.0} & {0.0} & {0.0}\\
     
     && {\methodName}  &\textbf{59.9} & \textbf{49.2} & \textbf{41.0} && \textbf{40.9} & \textbf{27.1} & \textbf{19.0} && \textbf{30.7} & \textbf{21.5} & \textbf{4.4} && \textbf{21.4} & \textbf{11.9} & \textbf{3.8}\\\cline{2-18}

     &\multirow{2}{4em}{\centering Rotten Tomatoes} &DAGER  &{10.7} & {1.4} & {1.3} && {6.5} & {0.0} & {0.0} && {4.6} & {0.0} & {0.0} && {0.0} & {0.0} & {0.0}\\
     
     && {\methodName}  &\textbf{39.0} & \textbf{21.3} & \textbf{18.4} && \textbf{37.0} & \textbf{24.5} & \textbf{16.0} && \textbf{20.5} & \textbf{16.4} & \textbf{5.9} && \textbf{16.7} & \textbf{7.1} & \textbf{1.9}\\\cline{2-18}

     &\multirow{2}{4em}{\centering MIMIC-III} &DAGER  &{5.6} & {1.9} & {1.1} && {3.0} & {0.0} & {0.0} && {0.0} & {0.0} & {0.0} && {0.0} & {0.0} & {0.0}\\
     
     && {\methodName}  &\textbf{21.5} & \textbf{15.2} & \textbf{13.1} && \textbf{20.8} & \textbf{9.7} & \textbf{7.4} && \textbf{9.2} & \textbf{1.7} & {0.0} && \textbf{8.2} & {0.0} & {0.0}\\
     \hline

\end{tabular}
}
\vspace{-2mm}
\label{tab: batch}
\end{table*}

\begin{table}[t]
\centering
\caption{Comparison of sequence reconstruction from gradients between \methodName and other baseline algorithms under two different PEFT methods with $b$ = 4 on CoLA dataset.}
\resizebox{\columnwidth}{!}{
\begin{tabular}{ccccccccccccccccc}
\hline
     \multirow{2}{4em}{\textbf{Model}}&\multirow{2}{4em}{\textbf{Method}} &\multicolumn{2}{c}{\textbf{w/o} \textbf{PEFT}} & &\multicolumn{2}{c}{\textbf{SLoRA}} & &\multicolumn{2}{c}{\textbf{FedAdapter}}\\\cline{3-4} \cline{6-7} \cline{9-10}
     
      & &R-1 &R-2 &&R-1 &R-2 &&R-1 &R-2 \\
     \hline
     \multirow{7}{4em}{BERT$_{Base}$} &DLG &35.3 &1.4 &&11.2 &1.2 &&9.8 &0.0\\
     &TAG    &35.3 &1.6  &&12.5 &1.3 &&11.3 &1.7\\
     & LAMP  &40.4 &6.4  &&15.1 &7.2 &&8.8 &2.6 \\
     &APRIL  &41.2 &7.8  &&17.3 &8.1 &&17.8 &5.7\\
     &FILM   &43.9 &11.4 &&21.3 &9.8 &&12.9 &3.6\\
     & DAGER &90.3 &85.6 &&61.2 &24.2 &&39.1 &15.8\\
     & \textbf{\methodName}  &\textbf{91.6} &\textbf{90.1} &&\textbf{72.5} &\textbf{42.1} &&\textbf{49.6} &\textbf{31.2}\\ 
    \hline
    
    \multirow{7}{4em}{GPT-2} & \text{DLG} & 21.1 & 1.3 & & 10.1 & 1.1 & & 8.8 & 1.4 \\
    & \text{TAG}   & 31.8 & 1.4  & & 11.3 & 1.2  & & 10.2 & 1.5 \\
    & \text{LAMP}  & 36.4 & 5.7  & & 13.6 & 6.5  & & 7.9 & 2.3 \\
    & \text{APRIL} & 37.1 & 6.9  & & 15.6 & 7.3  & & 16.0 & 5.1 \\
    & \text{FILM}  & 39.5 & 10.3 & & 19.2 & 8.8  & & 11.6 & 3.2 \\
    & \text{DAGER} & \textbf{100}  & \textbf{100}  & & 58.2 & 34.2 & & 45.3 & 31.8\\
    & \textbf{\methodName}  &{94.3} &{91.2} &&\textbf{61.9} &\textbf{40.3} &&\textbf{45.7} &\textbf{33.3}\\
     \hline

    \multirow{7}{4em}{\centering Llama-2 (7B)} &\text{DLG} & 14.8 & 0.9 & & 7.1 & 0.8 & & 3.2 & 0.0 \\
    & \text{TAG} & 22.3 & 1.0 & & 8.0 & 0.8 & & 7.2 & 1.0 \\
    & \text{LAMP} & 25.8 & 4.0 & & 9.7 & 4.6 & & 5.7 & 1.6 \\
    & \text{APRIL} & 31.4 & 5.8 & & 13.2 & 6.1 & & 13.5 & 4.3 \\
    & \text{FILM} & 33.5 & 8.7 & & 16.2 & 7.4 & & 9.8 & 2.7 \\
    & \text{DAGER} & \textbf{100}  & \textbf{100}  & & 48.7 & 31.2 & & 38.1 & 24.3\\
    & \textbf{\methodName}  & {96.6} & {93.1} & & \textbf{55.7} & \textbf{36.3} & & \textbf{41.1} & \textbf{30.0} \\

     \hline
\end{tabular}
}

\vspace{-2mm}
\label{tab: peft_results}
\end{table}

\subsection{Comparison against Baselines}
We first evaluated our method and several state-of-the-art algorithms across batch sizes ranging from 1 to 8 using BERT$_{Base}$ on three datasets. Among the LAMP variants, we selected $_{L_1 + L_2}$, as it consistently outperformed LAMP$_{cos}$.
The results, summarized in Table~\ref{tab: baselines}, show that our approach consistently exceeds baseline performance across different datasets and maintains its effectiveness even as the batch size grows. Notably, DAGER outperforms \methodName on CoLA and Rotten Tomatoes for batch sizes of 1 and 2, but its performance deteriorates at larger batch sizes due to the exploding search space and the difficulty of reconstructing duplicate tokens. Furthermore, the baseline methods exhibit a marked decline in performance when applied to next-word prediction tasks, which involve larger context lengths. Overall, analyzing the impact of batch size variations, we observe that attacks become increasingly challenging as the batch size grows. Nevertheless, \methodName outperforms all baselines in both tasks, achieving a performance improvement of 55\% to 160\% for batch size 8 in classification tasks, and an impressive 250\% increase in next-word prediction tasks. These results highlight the robustness and effectiveness of our method, particularly in handling complex scenarios where other methods falter.

\subsection{Performance under Larger Batch Sizes}
While most prior attacks degrade significantly at batch sizes as small as 8, we next assess \methodName’s performance alongside DAGER across three model architectures under larger batch sizes. Table~\ref{tab: batch} compares GPT-2 (decoder-based), BERT$_{base}$ (encoder-based), and T5$_{base}$ (encoder-decoder) on three different datasets, revealing that \methodName outperforms DAGER at larger batch sizes. As batch size grows, DAGER’s performance declines due to the increased number of tokens as well as the number of duplicate tokens, which confuses its search heuristics. By contrast, \methodName’s gradient inversion mechanism helps resolve this ambiguity. DAGER also struggles with encoder-based and encoder-decoder models due to their larger search space. Furthermore, both methods see reduced performance on the MIMIC-III dataset, regardless of model architecture, because of its longer context length. Even at batch sizes as large as $128$, \methodName can successfully recover tokens from gradients, underscoring its robustness and superior performance.

\begin{table}[t]
\centering
\caption{Text reconstruction comparison with $b$ = 4. Yellow marks tokens recovered correctly, green shows correct n-grams.}
\resizebox{\linewidth}{!}{
\begin{tabular}{cccc}
\hline
\parbox[c]{1.0cm}{\centering \textbf{Dataset}} &
\parbox[c]{1.7cm}{\centering \textbf{Method}} &
\parbox[c]{6.3cm}{\centering \textbf{Sequence}} &
\parbox[c]{1.6cm}{\centering \textbf{Entity-F1}($\uparrow$)} \tabularnewline
\hline

\multirow{18}{*}{\parbox[c]{1.0cm}{\centering Rotten\\Tomatoes}} &
\multirow{1}{*}{\parbox[c]{1.7cm}{\centering Reference}} &
\parbox[c]{6.3cm}{\rule{0pt}{5pt}some actors have so much charisma that you'd be happy to listen to them reading the phone book. hugh grant and sandra bullock are two such likable actors} &
\multirow{1}{*}{\parbox[c]{1.6cm}{\centering -}} \\
\cline{2-4}

& \multirow{1}{*}{\parbox[c]{1.7cm}{\centering TAG}} &
\parbox[c]{6.3cm}{\rule{0pt}{5pt}\hlone{like happy so listen the} the \hlone{are so listen.. sandra grant actors} like so\hlone{ism you} you} &
\multirow{1}{*}{\parbox[c]{1.6cm}{\centering 0.20}} \\
\cline{2-4}

& \multirow{1}{*}{\parbox[c]{1.7cm}{\centering LAMP}} &
\parbox[c]{6.3cm}{\rule{0pt}{5pt}\hlone{happy char sandra have such bull} the \hltwo{to listen} \hlone{char} sandra is \hlone{hu happy them be phone the grant} grant} &
\multirow{1}{*}{\parbox[c]{1.6cm}{\centering 0.27}} \\
\cline{2-4}

& \multirow{1}{*}{\parbox[c]{1.7cm}{\centering FILM}} &
\parbox[c]{6.3cm}{\rule{0pt}{5pt}\hltwo{grant and sandra} \hlone{phone to be} \hltwo{happy to} \hlone{char} happy phone \hlone{actors} \hltwo{actors have so much}} &
\multirow{1}{*}{\parbox[c]{1.6cm}{\centering 0.31}} \\
\cline{2-4}

& \multirow{1}{*}{\parbox[c]{1.7cm}{\centering BGP}} &
\parbox[c]{6.3cm}{\rule{0pt}{5pt}\hlone{be. the actors phone and} the \hlone{actors} phone be\hlone{. them so some d.} \hltwo{to} \hlone{have such} \hltwo{much charism.} \hlone{listen}. \hltwo{sandra bullock. }} &
\multirow{1}{*}{\parbox[c]{1.6cm}{\centering 0.32}} \\
\cline{2-4}

& \multirow{1}{*}{\parbox[c]{1.7cm}{\centering DAGER}} &
\parbox[c]{6.3cm}{\rule{0pt}{5pt}\hltwo{some actors have so} \hlone{charisma much} \hltwo{that you} \hlone{be'd happy them listen} \hltwo{the phone book} \hlone{grant hugh} \hltwo{and sandra bull are such} \hlone{actors likable}} &
\multirow{1}{*}{\parbox[c]{1.6cm}{\centering 0.55}} \\
\cline{2-4}

& \multirow{1}{*}{\parbox[c]{1.8cm}{\centering \textbf{\methodName}}} &
\parbox[c]{6.3cm}{\rule{0pt}{5pt}\hltwo{some actors have so much charisma that you’d be happy to listen} \hltwo{them read} \hlone{book and phone}. \hlone{hugh have} \hlone{grant}ing \hltwo{Sandra bullock} \hlone{that} \hltwo{are such} so \hltwo{likable}} &
\multirow{1}{*}{\parbox[c]{1.6cm}{\centering 0.71}} \\
\hline

\multirow{14}{*}{\parbox[c]{1.0cm}{\centering CoLA}} &
\multirow{1}{*}{\parbox[c]{1.7cm}{\centering Reference}} &
\parbox[c]{6.3cm}{\rule{0pt}{5pt}The more pictures of himself that John buys the more arrogant he becomes} &
\multirow{1}{*}{\parbox[c]{1.6cm}{\centering -}} \\
\cline{2-4}

& \multirow{1}{*}{\parbox[c]{1.7cm}{\centering TAG}} &
\parbox[c]{6.3cm}{\rule{0pt}{5pt}\hlone{buys him more the of} \hltwo{more pictures} \hltwo{be come} \hlone{arrogant John he is}} &
\multirow{1}{*}{\parbox[c]{1.6cm}{\centering 0.28}} \\
\cline{2-4}

& \multirow{1}{*}{\parbox[c]{1.7cm}{\centering LAMP}} &
\parbox[c]{6.3cm}{\rule{0pt}{5pt}\hlone{arrogant} \hltwo{the more pictures of} \hlone{more} \hltwo{John buys} \hlone{that the} \hltwo{he becomes} \hlone{himself} more} &
\multirow{1}{*}{\parbox[c]{1.6cm}{\centering 0.41}} \\
\cline{2-4}

& \multirow{1}{*}{\parbox[c]{1.7cm}{\centering FILM}} &
\parbox[c]{6.3cm}{\rule{0pt}{5pt}\hltwo{The more} \hlone{more John} \hltwo{pictures of him}, \hlone{the ant} \hltwo{more he becomes} \hlone{self arrog}} &
\multirow{1}{*}{\parbox[c]{1.6cm}{\centering 0.49}} \\
\cline{2-4}

& \multirow{1}{*}{\parbox[c]{1.7cm}{\centering BGP}} &
\parbox[c]{6.3cm}{\rule{0pt}{5pt}\hltwo{The more} \hlone{he buys}, \hltwo{the more ar} \hlone{John becomes ant} \hltwo{of himself}} &
\multirow{1}{*}{\parbox[c]{1.6cm}{\centering 0.49}} \\
\cline{2-4}

& \multirow{1}{*}{\parbox[c]{1.7cm}{\centering DAGER}} &
\parbox[c]{6.3cm}{\rule{0pt}{5pt}\hltwo{The more pictures of him} \hlone{ John self buy} he \hltwo{more arrogant he becomes}} &
\multirow{1}{*}{\parbox[c]{1.6cm}{\centering 0.71}} \\
\cline{2-4}

& \multirow{1}{*}{\parbox[c]{1.8cm}{\centering \textbf{\methodName}}} &
\parbox[c]{6.3cm}{\rule{0pt}{5pt}\hltwo{The more pictures of himself that John buys the more arrogant he becomes}} &
\multirow{1}{*}{\parbox[c]{1.6cm}{\centering 0.92}} \\
\hline

\multirow{16}{*}{\parbox[c]{1.0cm}{\centering MIMIC\\-III}} &
\multirow{1}{*}{\parbox[c]{1.7cm}{\centering Reference}} &
\parbox[c]{6.3cm}{\rule{0pt}{5pt}neonatology attending triage note baby olivero is a term male infant admitted to the nicu for sepsis evaluation} &
\multirow{1}{*}{\parbox[c]{1.6cm}{\centering -}} \\
\cline{2-4}

& \multirow{1}{*}{\parbox[c]{1.7cm}{\centering TAG}} &
\parbox[c]{6.3cm}{\rule{0pt}{5pt}\hlone{attending note} \hltwo{baby oliver} \hltwo{is a term male infant} \hlone{for} \hltwo{admitted to the} \hlone{evaluation}} &
\multirow{1}{*}{\parbox[c]{1.6cm}{\centering 0.26}} \\
\cline{2-4}

& \multirow{1}{*}{\parbox[c]{1.7cm}{\centering LAMP}} &
\parbox[c]{6.3cm}{\rule{0pt}{5pt}\hlone{the nicu age attending note oliver} is term \hltwo{is a term infant male admitted} male \hlone{for evaluation tri}} &
\multirow{1}{*}{\parbox[c]{1.6cm}{\centering 0.31}} \\
\cline{2-4}

& \multirow{1}{*}{\parbox[c]{1.7cm}{\centering FILM}} &
\parbox[c]{6.3cm}{\rule{0pt}{5pt}\hlone{neon admitted oliver} is sister\hlone{term infant ing the evaluation note baby to male is}} &
\multirow{1}{*}{\parbox[c]{1.6cm}{\centering 0.25}} \\
\cline{2-4}

& \multirow{1}{*}{\parbox[c]{1.7cm}{\centering BGP}} &
\parbox[c]{6.3cm}{\rule{0pt}{5pt}\hlone{neon oliver baby a attending} \hltwo{to the} \hlone{is} \hltwo{triage note} \hlone{sis for} the is to to} &
\multirow{1}{*}{\parbox[c]{1.6cm}{\centering 0.29}} \\
\cline{2-4}

& \multirow{1}{*}{\parbox[c]{1.7cm}{\centering DAGER}} &
\parbox[c]{6.3cm}{\rule{0pt}{5pt}\hlone{neon} \hltwo{triage note} \hlone{logy attending} \hltwo{baby oliver} \hltwo{is a} \hltwo{male infant admitted to the} \hlone{sis nic evaluation}} &
\multirow{1}{*}{\parbox[c]{1.6cm}{\centering 0.51}} \\
\cline{2-4}

& \multirow{1}{*}{\parbox[c]{1.8cm}{\centering \textbf{\methodName}}} &
\parbox[c]{6.3cm}{\rule{0pt}{5pt}\hltwo{neonatology} \hltwo{attending triage note baby oliver} is \hltwo{is a term male infant admitted to the nicu for sep} and is \hlone{evaluation}} &
\multirow{1}{*}{\parbox[c]{1.6cm}{\centering 0.68}} \\
\hline
\end{tabular}
}
\vspace{-2mm}
\label{tab: Hybrid}
\end{table}
\subsection{Performance under PEFT Methods}
PEFT methods significantly reduce the number of training parameters, which in turn decreases the amount of information stored in the gradients. We evaluated \methodName against various baselines using two different types of PEFT methods. The batch size was kept at 4, as we observed that baseline methods struggle to recover any meaningful information from PEFT gradients when the batch size exceeds 4. For our evaluation, we utilized the CoLA dataset because, as shown in Table~\ref{tab: peft_results}, all baselines perform well on this dataset (results for other datasets are provided in Appendix~\ref{sec: peft}).
From Table~\ref{tab: peft_results}, it is evident that the performance of all baselines declines when PEFT methods are applied. This is because PEFT gradients primarily capture the most significant directions of change, often overlooking finer details that would be present in full parameter gradients. Despite this general performance drop, \methodName consistently outperforms all baselines by a significant margin. Notably, \methodName is able to reconstruct up to $28\%$ longer subsequences on the BERT$_{Base}$ model compared to the baselines. This improvement is $17\%$ for GPT-2 and $23\%$ for Llama-2, demonstrating the method's superiority across different models.

\subsection{Sample Reconstructions}
We present sample sequence reconstructions of \methodName against various baselines on different datasets using the BERT\textsubscript{Base} model, with a batch size of 4, in Table~\ref{tab: Hybrid}. In the table, correctly reconstructed n-grams are highlighted in green, while correct unigrams are marked in yellow. Our method consistently generates more coherent and contextually accurate reconstructions, qualitatively surpassing the baselines. Particularly noteworthy is the performance on the MIMIC-III dataset, where all baselines struggled significantly with numerous uncommon medical terms, which other methods failed to recover accurately. However, \methodName excels in extracting these complex tokens. Moreover, while the baselines often falter in reconstructing the correct sequence order, our method nearly achieves perfect reconstruction, underscoring its superior performance and robustness.

\subsection{Impact of Sequence Order Calibration}
To assess the impact of the Sequence Recovery phase on reconstructed sequences, we conducted an ablation study across various datasets and LLMs with a batch size of 4. The results, presented in Table~\ref{tab: sequence}, show a significant improvement in sequence alignment following sequence order calibration. Notably, the ROUGE-1 score remains unchanged because it captures only the presence of individual words, without considering their order. However, there is a marked enhancement in the ROUGE-2 and ROUGE-L scores, which account for bigram and sequence-level structure, respectively, demonstrating that sequence order calibration effectively improves the accuracy of reconstructed sequences, aligning them more closely with the original training data.

\begin{table}[t]
\centering
\caption{Comparison of the reconstructed sequence before and after Sequence Order Calibration with $b$ = 4.}

\resizebox{\columnwidth}{!}{
\begin{tabular}{c c c c c c c c c}
\hline
\multirow{2}{*}{\parbox[c]{1.5cm}{\centering \textbf{Dataset}}} &
\multirow{2}{*}{\parbox[c]{1.4cm}{\centering \textbf{Model}}} &
\multicolumn{3}{c}{\parbox[c]{2.8cm}{\centering \textbf{Before Sequence\\Calibration}}} &
&
\multicolumn{3}{c}{\parbox[c]{2.8cm}{\centering \textbf{After Sequence\\Calibration}}} \\
\cline{3-5} \cline{7-9}

& & \textbf{R-1} & \textbf{R-2} & \textbf{R-L} & & \textbf{R-1} & \textbf{R-2} & \textbf{R-L} \\
\hline

\multirow{3}{*}{\parbox[c]{1.5cm}{\centering Rotten\\Tomatoes}}
& BERT$_{Base}$ & 76.9 & 8.3  & 22.0 & & 76.9 & 63.2 & 59.1 \\
& GPT-2         & 91.2 & 17.9 & 11.5 & & 91.2 & 86.8 & 71.8 \\
& Llama-2 (7B)  & 93.2 & 26.8 & 19.0 & & 93.2 & 88.4 & 83.1 \\
\hline

\multirow{3}{*}{\parbox[c]{1.5cm}{\centering CoLA}}
& BERT$_{Base}$ & 91.6 & 13.4 & 12.2 & & 91.6 & 90.1 & 86.9 \\
& GPT-2         & 94.3 & 21.2 & 24.5 & & 94.3 & 91.2 & 90.2 \\
& Llama-2 (7B)  & 96.6 & 23.5 & 26.4 & & 96.6 & 93.1 & 92.3 \\
\hline

\multirow{3}{*}{\parbox[c]{1.5cm}{\centering MIMIC-\\III}}
& BERT$_{Base}$ & 66.9 & 14.5 & 9.1  & & 66.9 & 53.2 & 59.1 \\
& GPT-2         & 88.1 & 13.9 & 11.4 & & 88.1 & 81.8 & 68.7 \\
& Llama-2 (7B)  & 89.3 & 32.7 & 16.3 & & 89.3 & 85.1 & 70.2 \\
\hline
\end{tabular}
}
\label{tab: sequence}
\vspace{-2mm}
\end{table}

\subsection{Robustness against Privacy Defenses}
We further assess the effectiveness of our attack against privacy defense mechanisms that incorporate Gaussian noise and/or gradient clipping, which are commonly employed in Differentially Private Stochastic Gradient Descent (DP-SGD)~\cite{geyer2017differentially}.
As is typical in privacy-preserving techniques, there is a trade-off: increasing noise enhances privacy but at the cost of reduced model accuracy. To explore this balance, we evaluated the performance of a fine-tuned BERT$_{Base}$ model on the Rotten Tomatoes dataset, carefully constraining the noise level ($\sigma$) and gradient clipping boundary ($\epsilon$) to minimize impact on model accuracy. Particularly, we refrained from testing noise levels above $10^{-5}$ and $\epsilon = 0.3$ due to the substantial accuracy degradation observed at these thresholds.

The results of our experiments, summarized in Table~\ref{tab: DP}, confirm that all baseline attacks have reduced reconstruction performance under these defenses. Notably, gradient clipping adversely affects all attacks, including \methodName, with DP further exacerbating this impact. However, it is noteworthy that even with the presence of privacy defenses, a substantial portion of the text remains recoverable. Moreover, \methodName significantly outperforms all baselines in performance, illustrating the robustness of our proposed attack against privacy defenses.

\begin{table}[t]
\centering
\caption{Comparison of different baselines against privacy defenses, with BERT$_{Base}$ $(b=4)$ on the Rotten Tomatoes dataset.}
\resizebox{\columnwidth}{!}{
\begin{tabular}{cccccccccccc}
\hline
     \multirow{3}{4em}{{Method}} &\multicolumn{3}{c}{\centering Additive Noise} &&\multicolumn{3}{c}{\centering Gradient Clipping} &&\multicolumn{3}{c}{\centering DP}\\

     &\multicolumn{3}{c}{\centering$\sigma = 10^{-5}$} & &\multicolumn{3}{c}{\centering$\epsilon=0.3$} & & \multicolumn{3}{c}{\centering$\sigma=10^{-5}$, $\epsilon=0.3$}\\\cline{2-4}\cline{6-8}\cline{10-12}
     
      & R-1 & R-2 & R-L & & R-1 & R-2 &R-L & & R-1 & R-2 &R-L\\
     \hline
    DLG    &13.4 &2.6 &0.4  &&10.7 &0.0 &0.3  &&9.4  &0.0  &0.3  \\
    TAG    &22.3 &6.2 &0.7  &&17.8 &3.0 &0.6  &&15.6 &1.3 &0.5  \\
    LAMP   &19.7 &6.0 &1.8  &&15.8 &2.8 &1.4  &&13.8 &1.2 &1.3  \\
    APRIL  &22.6 &11.2 &1.9  &&18.1 &3.6 &1.5  &&15.8 &1.9 &1.3  \\
    FILM   &24.9 &8.9 &2.1  &&19.9 &5.1 &1.7  &&17.4 &3.2 &1.5  \\
    BGP    &23.9 &9.4 &2.8  &&19.1 &5.5 &2.2  &&16.7 &3.6 &2.0  \\
    DAGER  &70.9 &{64.8} &41.4  &&54.9 &34.7 &23.9  &&36.9 &7.1 &5.9  \\
    \textbf{\methodName} &\textbf{71.5} &\textbf{68.1} &\textbf{58.2} &&\textbf{62.9} &\textbf{51.2} &\textbf{46.5} &&\textbf{59.1} &\textbf{47.2} &\textbf{41.8}\\
     \hline
\end{tabular}
}
\label{tab: DP}
\vspace{-2mm}
\end{table}

\begin{table}[t]
\centering
\caption{Effect of different model sizes on the performance of \methodName. Experiment with different sizes of Llama-2 and 3.1 models trained on CoLA dataset.}
\vspace{-0mm}
\resizebox{\columnwidth}{!}{
\begin{tabular}{cccccccccccc}
\hline
     \multirow{2}{4em}{{Model}} &\multicolumn{2}{c}{\centering $b$ = 1} &&\multicolumn{2}{c}{\centering $b$ = 8} &&\multicolumn{2}{c}{\centering $b$ = 32} &&\multicolumn{2}{c}{\centering $b$ = 64}\\\cline{2-3} \cline{5-6}\cline{8-9}\cline{11-12}
     
      & R-1 & R-2 & & R-1 & R-2 & & R-1 & R-2 & & R-1 & R-2\\
     \hline
     Llama-2 (7B)    &99.3 &99.1   &&93.6 &89.2  &&65.6 &48.1 &&51.7 &45.8\\
     Llama-2 (13B)   &100.0 &100.0 &&95.9 &91.3  &&71.1 &64.7 &&55.2 &47.1\\
     Llama-3.1 (8B)  &98.5 &95.2   &&94.2 &90.1  &&66.7 &49.3 &&52.3 &46.7\\
     Llama-3.1 (70B) &100.0 &100.0 &&98.2 &97.3  &&87.2 &84.1 &&69.1 &63.2\\
     \hline
\end{tabular}
}
\label{tab: model_size}
\vspace{-2mm}
\end{table}

\subsection{Effect of Model Size}
Prior work \cite{dengtaggradient, dimitrov2022lamp} suggests that model size strongly influences the extent of client information leakage. To explore this, in Table~\ref{tab: model_size} we present results of \methodName on different Llama-2 and Llama-3.1 model sizes (Experiments with GPT-2 models are provided in Appendix~\ref{sec: model_Size}) using the CoLA dataset and batch sizes up to $64$. Our findings reveal minimal performance differences across model sizes, although larger models produce more substantial gradients, which appear to enhance the effectiveness of our embedding regularization. This leads to more tokens being accurately reconstructed from gradients despite the increased batch size. Also note that Llama-3.1 has a larger vocabulary of 128, 256 tokens, but this does not affect the performance of \methodName.

\subsection{Runtime Analysis}
The BERT$_{Base}$ model experiments were among the most computationally intensive, particularly on the MIMIC-III dataset with a batch size of $b = 4$. \methodName required approximately $39$ hours to complete, balancing between performance and efficiency. TAG was the most time-efficient gradient-matching baseline on BERT, completing in just $13$ hours, while LAMP took $35$ hours. APRIL, benefiting from its analytical design, remained the fastest overall with $23$ hours. Among the search-based methods, BGP and FILM were more demanding, taking $60$ and $45$ hours, respectively, whereas DAGER was the slowest at $74$ hours.

In contrast, the experiments with GPT-2 and Llama-2 models were relatively faster. \methodName completed the MIMIC-III runs in $9$ hours on GPT-2 and $14$ hours on Llama-2 (7B) for $b = 4$. When the batch size was increased to $b = 8$ on Llama-2 (7B), the runtime rose to $19$ hours. This reflects the added cost of the Sequence Order Calibration phase, which contributes to increased computation time at higher batch sizes. Despite this, \methodName maintains a favorable efficiency-performance trade-off, especially compared to search-based methods.
%-------------------------------------------------------------------------------

\section{Related Work}
Data reconstruction attacks in federated learning (FL) can be broadly categorized into \textit{Gradient-matching-based} and \textit{Gradient-analytic-based} methods. \textit{Gradient-matching-based} attacks rely on optimization techniques to iteratively match the shared gradients with reconstructed data, often using loss functions and priors to refine recovery. \textit{Gradient-analytic-based} methods, which also encompass search-based approaches, extract data directly from gradients by leveraging their inherent properties, such as linear relationships or rank deficiencies. Search-based techniques, such as beam search or token alignment, are often employed in this category to assemble data representations from gradient information. These two categories have been explored extensively in recent works:

%-------------------------------------------------------------------------------
\noindent\textbf{Gradient-Matching-based Attacks.}
The concept of data reconstruction attacks, introduced by~\cite{zhu2019deep}, showed that an adversary could exploit shared gradients to reconstruct a client’s private data samples. Subsequent studies~\cite{zhao2020idlg, zhu2020r} reinforced the notion that gradients can still leak information, challenging the fundamental privacy assumptions of FL. In the image domain, recent works~\cite{yin2021see, jeon2021gradient, li2022auditing} have achieved highly accurate data reconstruction from gradients, highlighting the significant privacy risks in FL.

In the text domain, gradient inversion faces additional challenges due to the discrete and sequential nature of text data. DLG~\cite{zhu2019deep} was the first method to reconstruct text from gradients, but its performance degraded with longer sequences and larger batch sizes. TAG~\cite{dengtaggradient} improved on DLG by adding an $L_1$ penalty to stabilize optimization. Similarly, LAMP~\cite{dimitrov2022lamp} introduced cosine similarity for reconstruction loss, added embedding regularization, and incorporated a language model prior to guide reconstruction towards more natural text. Despite these advancements, most gradient inversion attacks require small batch sizes and complete access to model parameters, limiting their practical applicability.

\noindent\textbf{Gradient-Analytic-based Attacks.}
Analytic-based attacks recover the data directly from the gradient itself without relying on optimization or gradient matching. For example, FILM~\cite{gupta2022recovering} recovers token representations by observing the gradients of embedding layers and then assembles them into sequences via beam search. However, a key limitation is its assumption that embedding layers are actively trained—an assumption that rarely holds in federated fine-tuning, where embedding layers are typically frozen. APRIL~\cite{lu2022april} takes a step further by recovering positional embeddings, thereby achieving exact gradient leakage in transformers; nonetheless, it relies on a batch size of one and full access to positional embedding gradients, making it less feasible for large-scale federated environments. BGP~\cite{li2023beyond} introduces a two-stage privacy attack by targeting the Pooler layer. In the first stage, BGP recovers intermediate feature directions to serve as supervisory signals; then, in the second stage, it combines these signals with gradient inversion and prior knowledge to recover training data. However, BGP depends on specific architectural features, limiting its applicability to models with Pooler layers, which degrades its performance with larger batch sizes and interdependent features. DAGER~\cite{petrov2024dager}, while not restricted to analyzing embedding gradients alone, exploits the rank-deficiency of gradients to reconstruct text tokens even under parameter-efficient fine-tuning (PEFT) conditions, thus circumventing some of the constraints of earlier embedding-based methods. That said, DAGER’s search-based methodology can incur substantial computational overhead on larger batch sizes, highlighting the persistent trade-offs between attack effectiveness, model complexity, and scalability.

\section{Conclusion}
%-------------------------------------------------------------------------------
In this paper, we introduced \methodName, a scalable approach for reconstructing private text data from gradients in LLMs, even under the constraints of various PEFT methods. By exploiting the linear dependency between embedding vectors and gradient columns, and incorporating a sequence order calibration method, \methodName accurately recovers tokens and preserves their correct sequence. Our experiments show that \methodName consistently outperforms prior attacks across diverse datasets and model architectures, proving robustness against challenges like larger batch sizes, longer sequences, and different PEFT methods. These findings expose critical privacy risks associated with gradients in federated LLMs, emphasizing the urgent need for stronger privacy-preserving measures to protect sensitive text data.

\section{Acknowledgement}
We would like to thank our anonymous reviewers and shepherd for their insightful feedback. This work is supported in part by NSF CNS-2114161, ECCS-2132106, CBET-2130643, and CNS-2403529.
\bibliographystyle{ieeetr} 
\bibliography{ref}

\appendix
\setcounter{theorem}{0}
\setcounter{lemma}{0}
\section{Appendix}
\subsection{Theorems and Proofs}\label{sec: theorem}
\begin{theorem}
As the attention layers of the transformer are linear layers (e.g., $\bm{Y} = \bm{ZW} + \bm{B}$, where $\bm{Y}$ represents the output, $\bm{Z}$ denotes the input, $\bm{W}$ is the layer's weight matrix, and $\bm{B}$ is the bias term), the gradients of the loss $\mathcal{L}$ w.r.t $\bm{W}$ for a linear layer can be expressed as:
\begin{equation}\label{eq: gradient}
    \frac{\partial \mathcal{L}}{\partial \bm{W}} = \bm{Z^T} \frac{\partial \mathcal{L}}{\partial \bm{Y}}.
\end{equation}
\end{theorem}

\begin{proof}
     To find the gradient of $\mathcal{L}$ with respect to the weight matrix $\bm{W}$, we utilize the chain rule from matrix calculus. Applying the chain rule yields:
\begin{equation}
\frac{\partial \mathcal{L}}{\partial \bm{W}} = \frac{\partial \mathcal{L}}{\partial \bm{Y}} \frac{\partial \bm{Y}}{\partial \bm{W}},
\end{equation}
where $\frac{\partial \mathcal{L}}{\partial \bm{Y}}$ denotes the gradient of the loss with respect to the output $\bm{Y}$. Given that $\bm{Y} = \bm{ZW} + \bm{B}$, and the bias term $\bm{B}$ is independent of $\bm{W}$, the partial derivative $\frac{\partial \bm{Y}}{\partial \bm{W}}$ is determined solely by the derivative of the term $\bm{ZW}$. 
By ignoring the bias term $\bm{B}$ for each element ${Y}_{ij}$ of $\bm{Y}$, we have:
\begin{equation}
{Y}_{ij} = \sum_{k=1}^{d} {Z}_{ik} {W}_{kj}.
\end{equation}

By taking the partial derivative of $Y_{ij}$ with respect to an element $W_{ab}$ of $\bm{W}$, we have:
\begin{equation}
\frac{\partial Y_{ij}}{\partial W_{ab}} = \frac{\partial}{\partial W_{ab}} \left(\sum_{k=1}^{d} Z_{ik} W_{kj}\right) = Z_{ia} \delta_{jb},
\end{equation}
where $\delta_{jb}$ is the Kronecker delta, which is 1 if $j = b$ and 0 otherwise. Therefore, the derivative $\frac{\partial \bm{Y}}{\partial \bm{W}}$ is a tensor, but it can be represented compactly in matrix form as:
\begin{equation}
   \frac{\partial \bm{Y}}{\partial \bm{W}} = \bm{Z}^T.
\end{equation}
This derivative tells us how each element of $\bm{W}$ affects the corresponding output $\bm{Y}$, and substituting the result into the chain rule, we can obtain:
\begin{equation}
\frac{\partial \mathcal{L}}{\partial \bm{W}} = \bm{Z}^T \frac{\partial \mathcal{L}}{\partial \bm{Y}} .
\end{equation}
\end{proof}

\begin{theorem}\label{theorem:2}
Under the assumption that 
%the total number of tokens in the batch is smaller than the embedding dimension ($t < d$) and 
$\frac{\partial \mathcal{L}}{\partial \bm{W}}$ is rank-deficient, %with rank at most $t$, 
$\bm{Z^T}$ is a linear combination of the columns of $\frac{\partial \mathcal{L}}{\partial \bm{W}}$.
\end{theorem}

\begin{proof}
     We perform QR decomposition of $\frac{\partial {\mathcal{L}}}{\partial \bm{W}}$:
\begin{equation}
    \frac{\partial \mathcal{L}}{\partial \bm{W}} = \bm{QR},
\end{equation}
where $\bm{Q}$ is an orthogonal matrix and $\bm{R}$ is an upper triangular matrix. Because $\frac{\partial \mathcal{L}}{\partial \bm{W}}$ is rank deficient, some of the diagonal elements of $\bm{R}$ will be zero.
Substituting the QR decomposition into Equation~\ref{eq: gradient}, we can obtain:
\begin{equation}
    \bm{QR} = \bm{Z}^T\frac{\partial \mathcal{L}}{\partial \bm{Y}}.
\end{equation}
To isolate $\bm{Q}$, we use the Moore-Penrose pseudoinverse \cite{moore1920reciprocal}, which provides a generalized inverse for matrices, particularly useful when dealing with rank-deficient matrices. Applying the pseudoinverse, we can get:
\begin{equation}
    \bm{Q} = \bm{Z}^T \bm{R}^\dagger\frac{\partial {\mathcal{L}}}{\partial \bm{Y}} ,
\end{equation}
where $\bm{R}^\dagger$ denotes the pseudoinverse of $\bm{R}$. Since $\bm{Q}$ is an orthogonal matrix, its columns form an orthonormal basis, and $\bm{Z}^T\bm{R}^\dagger$ implies that $\bm{Z}^T$ is being scaled and rotated by $\bm{R}^\dagger$. Since $\bm{Q}$ is orthogonal and $\frac{\partial {\mathcal{L}}}{\partial \bm{W}}$ is rank deficient, $\bm{R}$ has rows or columns that are zero, making $\bm{R}\dagger$ only act on the non-zero parts. Thus, $\bm{Z}^T$ must lie within the column space of $\bm{Q}$, which is also the column space of $\frac{\partial {\mathcal{L}}}{\partial \bm{W}}$.
\end{proof}

\begin{lemma} \label{lemma:}
Let $\{\mathbf{g}_p\}_{p=1}^n$ be a sequence of vectors in $\mathbb{R}^d$. Let the total sum be $S \;=\; \sum_{p=1}^n \mathbf{g}_p$. For any \( k < n \), define the partial sum $S_{\le k} \;=\; \sum_{p=1}^k \mathbf{g}_p$. Then,
\[
S_{\le k} \cdot S
\;=\;
\bigl\|S_{\le k}\bigr\|^2
\;+\;
S_{\le k} \,\cdot\, \sum_{p=k+1}^n \mathbf{g}_p.
\]
In particular, if the remaining vectors \( \{\mathbf{g}_p\}_{p=k+1}^n \) are not strongly negatively correlated with \( S_{\le k} \), then the dot product \( S_{\le k} \cdot S \) will remain large.
\end{lemma}

\begin{proof}
    By definition,
\[
S_{\le k} \,\cdot\, S
\;=\;
\left(\sum_{p=1}^k \mathbf{g}_p\right)
\cdot
\left(\sum_{p=1}^n \mathbf{g}_p\right).
\]
Distribute the dot product:
\[
=\;
\sum_{p=1}^k \sum_{p=1}^n \mathbf{g}_p \cdot \mathbf{g}_p
\;=\;
\sum_{p=1}^k \mathbf{g}_p \cdot \mathbf{g}_p
\;+\;
\sum_{p=1}^k \sum_{p=k+1}^n \mathbf{g}_p \cdot \mathbf{g}_p.
\]
Notice that
$\sum_{p=1}^k \mathbf{g}_p \cdot \mathbf{g}_p = \|S_{\le k}\|^2$ minus the cross terms within $\mathbf{g}_1, \dots, \mathbf{g}_k$. However, rearranging properly and grouping vectors gives us exactly
\begin{equation}
\|S_{\le k}\|^2
\;+\;
S_{\le k} \,\cdot\, \sum_{p=k+1}^n \mathbf{g}_p.
\end{equation}
Hence the stated identity follows.
\end{proof}

\begin{theorem}\label{theorem}
Consider a full sequence $\textbf{x} = (x_1, x_2, \dots, x_n)$ and two partial sequences of length $k$: A correctly ordered partial sequence $x_{\le k}^{\text{(correct)}}, e.g., (x_1, x_2, \dots, x_k)$, and a misordered partial sequence $x_{\le k}^{\text{(wrong)}}$, e.g., $(x_1, x_3, \dots)$. Let 
\begin{equation}
    \nabla_\theta^{\text{full}}:==-\sum_{t=1}^n \nabla_\theta \log p_\theta\bigl(x_t \mid x_{<t}\bigr),
\end{equation}
\begin{equation}
    \nabla_\theta^{\text{correct}} := -\sum_{t=1}^k \nabla_\theta \log p_\theta\bigl(x_t \mid x_{<t}\bigr),
\end{equation}

\begin{equation}
\begin{split}
\nabla_\theta^{\text{wrong}}:= -\Bigl[\nabla_\theta \log p_\theta(x_1) + \nabla_\theta \log p_\theta(x_3 \mid x_1) + \dots \\
+ \nabla_\theta \log p_\theta(x_t \mid x_{t-1})\Bigr],
\end{split}
\end{equation}

Then,
\[ \nabla_\theta^{\text{correct}} \cdot \nabla_\theta^{\text{full}} \;>\; \nabla_\theta^{\text{wrong}} \cdot \nabla_\theta^{\text{correct}},
\]
indicating that the gradient for the correctly ordered partial sequence is more aligned with the gradient of the full sequence than the gradient for the misordered partial sequence.
\end{theorem}

\begin{proof}
By Decomposing $\nabla_\theta^{\text{full}}$ we observe that  
   \begin{equation}
       \nabla_\theta^{\text{full}}
   \;=\;
   \nabla_\theta^{\text{correct}} - \sum_{t=k+1}^n \nabla_\theta \log p_\theta(x_t \mid x_{<t}).
   \end{equation}
Hence $\nabla_\theta^{\text{correct}}$ is the partial sum of the first $k$ gradient terms in $\nabla_\theta^{\text{full}}$.

Let $g_t = - \nabla_\theta \log p_\theta(x_t \mid x_{<t})$, applying Lemma~\ref{lemma:}, we get
\begin{equation}
    \nabla_\theta^{\text{correct}}\cdot \nabla_\theta^{\text{full}} = \bigl\|\nabla_\theta^{\text{correct}}\bigr\|^2 + \nabla_\theta^{\text{correct}}\cdot \sum_{t=k+1}^n g_t,
\end{equation}
Provided the vectors ${g_t}$ are not negatively correlated, $\nabla_\theta^{\text{correct}}\cdot \nabla_\theta^{\text{full}}$ tends to be substantial. 

In case of $\nabla_\theta^{\text{wrong}}$ the second term is $\nabla_{\theta}\log p_{\theta}(x_3\mid x_1)$ disrupts the natural alignment with the subsequent gradients for preceding tokens. Hence,

\begin{equation}
    \nabla_\theta^{\text{wrong}}\cdot \nabla_\theta^{\text{full}} = g_1 + \Tilde{g}_2 + \sum_{t=3}^k g_t,
\end{equation}
where, $\Tilde{g}_2 = - \nabla_\theta \log p_\theta(x_3 \mid x_1)$ This $\Tilde{g}_2$ typically points in a direction that reduces alignment with the true partial sum of $\sum_{t=3}^k g_t$. Therefore, under these assumptions
\begin{equation}
    \nabla_\theta^{\text{correct}}\cdot \nabla_\theta^{\text{full}} > \nabla_\theta^{\text{wrong}}\cdot \nabla_\theta^{\text{full}}
\end{equation}
\end{proof} 

\vspace{-1mm}
\subsection{Baselines and Parameter Setting} \label{app: baseline}
We benchmark our attack against seven baselines: DLG~\cite{zhao2020idlg}, TAG~\cite{dengtaggradient}, LAMP~\cite{dimitrov2022lamp}, APRIL~\cite{lu2022april}, FILM~\cite{gupta2022recovering}, BGP~\cite{li2023beyond} and DAGER~\cite{petrov2024dager}.
Among these, DLG, TAG, LAMP, and FILM are SOTA optimization-based data reconstruction attacks. APRIL, although originally proposed for vision transformers, can be adapted for NLP tasks as an exact data reconstruction attack. BGP employs a hybrid approach by combining an analytics-based attack with an optimization-based attack, and DAGER is an analytics search-based method.

In our experiments, we configure DLG, TAG, and LAMP to run for $30$ iterations, each with $75$ continuous optimization steps. For LAMP, we add another $200$ discrete optimization steps to minimize the combined reconstruction loss and perplexity. DAGER uses a span check acceptance threshold of $10^{-5}$ in the first layer and $10^{-3}$ in the second layer, along with a rank truncation of $\Delta=20$. For FILM, we perform $90,000$ iterations starting from an initial learning rate of $10^{-5}$, which is then linearly decayed. When using BGP, we set the feature match loss margin scale to $0.1$ and rely on a grid search, following the strategy outlined in its original paper. Regarding our own method’s hyperparameters, we run the continuous optimization for 100 iterations with $\gamma = 0.4$, and for the PEFT scenario specifically, we set $\gamma = 0.35$ and $\eta = 0.2$. Unless otherwise specified, the context length is $32$ for the CoLA dataset, $128$ for Rotten Tomatoes, and $256$ for MIMIC-III. For SLoRA, we employ a rank of $32$, and for FedAdapter, we choose a bottleneck dimension of $32$.
\subsection{Additional Evaluation Details}
\textbf{Large Language Models.}
In this work, we rigorously explore the effectiveness of our attack across various model architectures. As a starting point, we use the GPT-2 base variant, which consists of a 12-layer transformer architecture, a 768-dimensional hidden state, 12 attention heads, and a feed-forward filter with a size of 3072. We specifically chose the GPT-2 base variant for its balance between model complexity and computational efficiency, making it a standard benchmark in many language modeling tasks.

To further validate the generality and effectiveness of our attack, we extend our analysis to state-of-the-art models, such as Llama 2 (7B), BERT${Base}$, and T5${Base}$. Llama 2 (7B), with its 7 billion parameters, represents a cutting-edge decoder model, while BERT${Base}$, a 12-layer transformer with a 768-dimensional hidden state, serves as a representative encoder-only model. For experiments involving the MIMIC-III dataset, we utilize the BioBERT~\cite{lee2020biobert} version of the BERT${Base}$ model. Additionally, we evaluate our method on T5$_{Base}$, a 12-layer encoder-decoder transformer with a 768-dimensional hidden state and a feed-forward dimension of 2048. T5 serves as a representative of encoder-decoder architectures, enabling us to assess our method's performance across the full spectrum of model types.

All models used in this work were sourced in their pre-trained form from HuggingFace~\cite{wolf2019huggingface}, ensuring consistency with standard practices in language modeling research.

\textbf{Hardware Details.}
We implemented~\methodName in PyTorch~\cite{paszke2019pytorch} and conducted all experiments using NVIDIA A100 Tensor Core GPUs. For tests on the Llama-2 (7B) and GPT-2 architectures, we utilized two NVIDIA A100 GPUs, which offer 40 GB of memory. All other experiments were executed on a single GPU. The required RAM varied widely across experiments, ranging from 16 GB to 250 GB, depending on the model, batch size, and sequence length used.

\begin{figure}[t]
    \centering
    \includegraphics[width=0.7\linewidth]{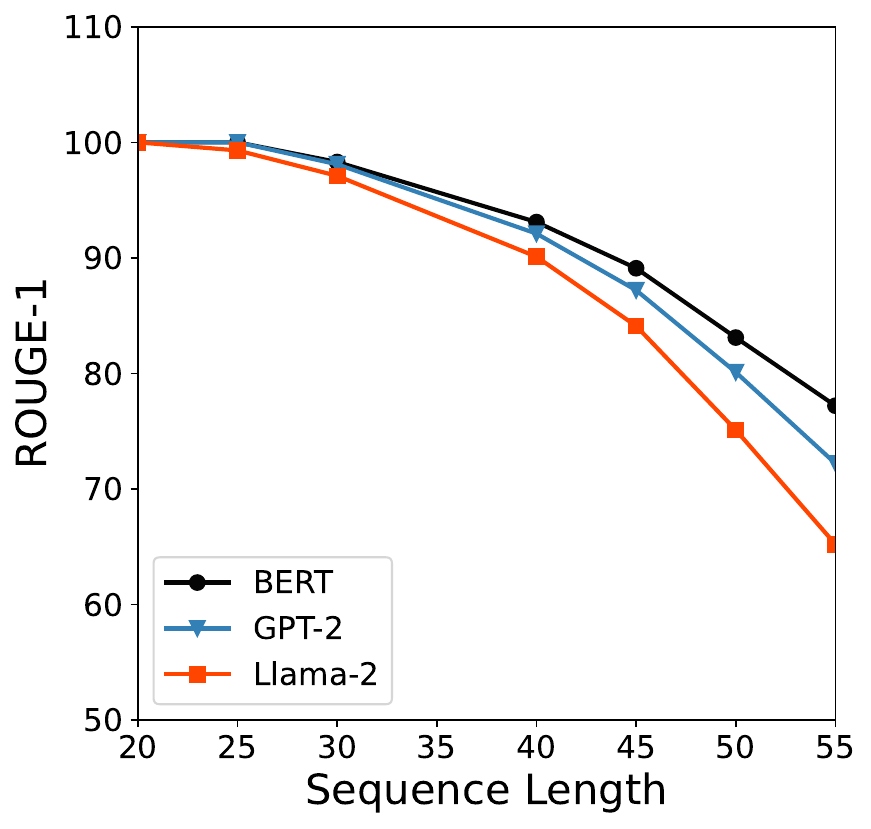}
    \caption{Impact of sequence length on the reconstruction efficiency.}
    \label{fig:sequence length}
\end{figure}

\subsection{Impact of Sequence Length}
As sequence length increases, the difficulty of recovering tokens from the gradient also grows. This challenge arises because longer sequences produce a more diffuse gradient signal, making it harder to isolate the contribution of individual tokens. In gradient inversion attacks, the information embedded in the gradient becomes less distinct as the number of tokens increases, leading to a degradation in token recovery performance. This phenomenon underscores the complexity of extracting specific token-level information when dealing with longer sequences, where the gradient's sensitivity to each token diminishes.

To evaluate how \methodName performs under these challenging conditions, we conducted experiments using the Rotten Tomatoes dataset \cite{socher2013recursive} with a batch size of $b=1$ and the BERT$_{Base}$ model under different sequence lengths. The results, as shown in Figure~\ref{fig:sequence length}, indicate that \methodName is highly effective in recovering tokens when the sequence length is up to $40$, with nearly all tokens being accurately recovered. Even as the sequence length increases, our attack continues to perform robustly, recovering more than $70\%$ of tokens for sequence lengths up to $55$.

\begin{table}[t]
\centering
\caption{Comparison of the reconstructed sequence using different loss functions with $b$ = 4.}
\resizebox{\columnwidth}{!}{
\begin{tabular}{c c c c c c c c}
\hline
\multirow{2}{*}{\parbox[c]{1.5cm}{\centering Dataset}} &
\multirow{2}{*}{\parbox[c]{1.4cm}{\centering Model}} &
\multicolumn{2}{c}{\parbox[c]{1.8cm}{\centering L2}} &
\multicolumn{2}{c}{\parbox[c]{1.8cm}{\centering L1 + L2}} &
\multicolumn{2}{c}{\parbox[c]{3.0cm}{\centering Weighted Layer-wise\\Cosine Similarity}} \\
\cline{3-4} \cline{5-6} \cline{7-8}

& & R-1 & R-2 & R-1 & R-2 & R-1 & R-2 \\
\hline

\multirow{3}{*}{\parbox[c]{1.5cm}{\centering Rotten\\Tomatoes}}
& BERT$_{Base}$ & 68.9 & 50.4 & 71.5 & 57.8 & \textbf{76.9} & \textbf{63.2} \\
& GPT-2         & 67.5 & 50.1 & 70.2 & 55.5 & \textbf{91.2} & \textbf{86.8} \\
& Llama-2 (7B)  & 63.8 & 45.7 & 66.7 & 52.1 & \textbf{93.2} & \textbf{88.4} \\
\hline

\multirow{3}{*}{\parbox[c]{1.5cm}{\centering CoLA}}
& BERT$_{Base}$ & 66.9 & 50.5 & 69.5 & 59.2 & \textbf{91.6} & \textbf{90.1} \\
& GPT-2         & 58.4 & 38.1 & 60.8 & 44.6 & \textbf{94.3} & \textbf{91.2} \\
& Llama-2 (7B)  & 51.8 & 35.1 & 54.3 & 39.5 & \textbf{96.6} & \textbf{93.1} \\
\hline

\multirow{3}{*}{\parbox[c]{1.5cm}{\centering MIMIC-\\III}}
& BERT$_{Base}$ & 53.5 & 42.6 & 58.2 & 46.3 & \textbf{66.9} & \textbf{53.2} \\
& GPT-2         & 51.8 & 40.6 & 56.3 & 44.2 & \textbf{88.1} & \textbf{81.8} \\
& Llama-2 (7B)  & 46.6 & 36.1 & 50.7 & 39.2 & \textbf{89.3} & \textbf{85.1} \\
\hline
\end{tabular}
}
\label{tab: distance_metric}
\end{table}

\subsection{Impact of Loss Function}
At the core, \methodName operates as a gradient matching method designed to closely replicate the shared gradients, which are crucial for effective reconstruction. The choice of distance metric plays a pivotal role in determining the quality of this reconstruction, as it directly impacts how well the method can align the gradients and, consequently, recover the original tokens.  

To thoroughly assess the impact of different distance metrics, we evaluated \methodName using three distinct metrics. The results, summarized in Table~\ref{tab: distance_metric}, reveal significant variations in performance depending on the metric used. Among the metrics tested, the Weighted Layer-wise Cosine Similarity emerged as the most effective, enabling \methodName to recover approximately 20\% more tokens compared to the other metrics. This substantial improvement can be attributed to the ability of the Weighted Layer-wise Cosine Similarity to better capture the nuanced relationships between layers and gradients. By assigning appropriate weights to different layers, this metric ensures that more critical layers contribute proportionally to the gradient alignment, leading to more accurate token recovery.

\subsection{Additional Results of Performance under PEFT Methods}\label{sec: peft}

To further evaluate the robustness of \methodName under PEFT-induced gradient sparsity, we conduct experiments on the Rotten Tomatoes and MIMIC-III datasets using two different PEFT methods: SLoRA and FedAdapter. We benchmark our method against LAMP and DAGER across BERT$_{Base}$ and GPT-2 models. In these experiments, we measure reconstruction quality using ROUGE-1 and ROUGE-2 scores, with a fixed batch size of 4 to ensure a fair comparison across all methods.

From Table~\ref{tab: peft_results_2}, we observe that the performance of both LAMP and DAGER deteriorates significantly under PEFT settings. This drop is especially pronounced under FedAdapter, where gradients become highly sparse. In contrast, \methodName consistently achieves the highest ROUGE scores across all settings. On the Rotten Tomatoes dataset, \methodName improves ROUGE-1 by up to $24\%$ on BERT${Base}$ and $11\%$ on GPT-2 under SLoRA, compared to DAGER. On the MIMIC-III dataset, these improvements are even more substantial, with up to $28\%$ and $22\%$ gains in ROUGE-1 for BERT${Base}$ and GPT-2 respectively, under FedAdapter.

\begin{table}[t]
\centering
\caption{Comparison of sequence reconstruction from gradients between \methodName and other baseline algorithms under two different PEFT methods with $b$ = 4 on different datasets.}
\resizebox{\columnwidth}{!}{
\begin{tabular}{cccccccccccccccccc}
\toprule
      \multirow{2}{4em}{\textbf{Dataset}} &\multirow{2}{4em}{\textbf{Model}}&\multirow{2}{4em}{\textbf{Method}} &\multicolumn{2}{c}{\textbf{w/o} \textbf{PEFT}} & &\multicolumn{2}{c}{\textbf{SLoRA}} & &\multicolumn{2}{c}{\textbf{FedAdapter}}\\\cline{4-5} \cline{7-8} \cline{10-11}
     
      && &R-1 &R-2 &&R-1 &R-2 &&R-1 &R-2 \\
     \hline
     \multirow{6}{4em}{Rotten Tomatoes}&\multirow{3}{4em}{BERT$_{Base}$}& LAMP  &24.6 &2.3  &&8.2 &0.0 &&2.4 &0.0 \\
     && DAGER &64.2 &42.8 &&37.2 &24.5 &&32.1 &22.9\\
     && \textbf{\methodName}  &\textbf{76.9} &\textbf{63.2} &&\textbf{61.5} &\textbf{50.5} &&\textbf{46.1} &\textbf{38.5}\\ 
    \cline{2-11}
    
    &\multirow{3}{4em}{GPT-2} & \text{LAMP}  & 6.3 & 1.2  &&0.0 &0.0   &&0.0  &0.0  \\
    && \text{DAGER}           & \textbf{97.2}  & \textbf{91.2}  &&58.3 &44.7 &&43.7 &31.1 \\
    && \textbf{\methodName}   &91.2 &86.8 &&\textbf{73.1} &\textbf{69.7} &&\textbf{54.7} &\textbf{53.1}\\
     \cline{2-11}

     \hline
     
     \multirow{6}{4em}{MIMIC-III}&\multirow{3}{4em}{BERT$_{Base}$}& LAMP  &7.4 &1.5 &&5.2 &1.2 &&4.8 &0.0 \\
     && DAGER &42.5 &31.9 &&25.5 &19.4 &&21.1 &16.2\\
     && \textbf{\methodName}  &\textbf{66.9} &\textbf{53.2} &&\textbf{53.6} &\textbf{42.4} &&\textbf{40.1} &\textbf{32.1}\\ 
    \cline{2-11}
    
    &\multirow{3}{4em}{GPT-2} & \text{LAMP}  &3.4 &0.8 &&2.1 &0.0 &&2.3 &0.0 \\
    && \text{DAGER}          &\textbf{90.2} &\textbf{82.1} &&54.1 &48.2 &&43.1 &23.5\\
    && \textbf{\methodName}  &88.1 &81.8 &&\textbf{66.1} &\textbf{62.7} &&\textbf{54.1} &\textbf{46.7}\\

    \bottomrule
\end{tabular}
}
\vspace{-3mm}
\label{tab: peft_results_2}
\end{table}

\subsection{Additional Results for Effect of Model Size}\label{sec: model_Size}
We report additional results in Table~\ref{tab: model_size_1} using GPT-2 models of increasing size (Small, Medium, and Large) on the CoLA dataset. From the table, we can see that the performance degrades as batch size increases, consistent with gradient dilution effects. Interestingly, smaller GPT-2 models retain higher reconstruction fidelity under larger batches compared to larger models. For instance, at $b=64$, the GPT-2 Small model achieves $86.3$ R-1, while GPT-2 Large drops to $58.7$. This supports prior findings~\cite{dengtaggradient, dimitrov2022lamp} that larger models diffuse gradient information more broadly, making reconstruction harder. Nonetheless, \methodName remains effective across all model sizes, showcasing its resilience even as gradient quality deteriorates in deeper architectures.

\begin{table}[t]
\centering
\caption{Effect of different model sizes on the performance of \methodName. Experiment with different sizes of GPT-2 models trained on the CoLA dataset.}
\resizebox{\columnwidth}{!}{
\begin{tabular}{cccccccccccc}
\hline
     \multirow{2}{4em}{{Model}} &\multicolumn{2}{c}{\centering $b$ = 1} &&\multicolumn{2}{c}{\centering $b$ = 8} &&\multicolumn{2}{c}{\centering $b$ = 32} &&\multicolumn{2}{c}{\centering $b$ = 64}\\\cline{2-3} \cline{5-6}\cline{8-9}\cline{11-12}
     
      & R-1 & R-2 & & R-1 & R-2 & & R-1 & R-2 & & R-1 & R-2\\
     \hline
     GPT-2 (Small)    &100.0 &100.0   &&94.1 &90.8  &&91.2 &87.6 &&86.3 &71.6\\
     GPT-2 (Medium)   &100.0 &100.0   &&93.7 &90.1  &&85.3 &81.1 &&77.9 &67.6\\
     GPT-2 (Large)    &100.0 &100.0   &&91.2 &87.2  &&75.8 &61.9 &&58.7 &45.7\\
     \hline
\end{tabular}
}
\label{tab: model_size_1}
\end{table}

\end{document}